\documentclass[runningheads]{llncs}
\usepackage{graphicx}
\usepackage{amsmath}
\usepackage{amstext}
\usepackage{amssymb}
\usepackage{amsfonts}
\usepackage{xspace}
\usepackage{graphicx}
\usepackage{url}
\usepackage{subfigure}
\usepackage{mdwlist}
\usepackage{courier}
\usepackage{lscape}
\usepackage{comment}
\usepackage{wrapfig}
\usepackage{multirow}
\usepackage{bussproofs}
\usepackage{pdftricks}
\usepackage{subfigure}
\begin{psinputs}
\usepackage{pstricks}
\usepackage{color}
\usepackage{pstcol}
\usepackage{pst-plot}
\usepackage{pst-tree}
\usepackage{pst-eps}
\usepackage{multido}
\usepackage{pst-node}
\usepackage{pst-eps}
\end{psinputs}
\usepackage{cite}
\makeatletter
\newif\if@restonecol
\makeatother

\usepackage{footmisc}
\usepackage{balance}
\usepackage{enumerate}
\usepackage[normalem]{ulem}
\usepackage{hyperref}
\usepackage{alltt}

\usepackage{epsfig}
\usepackage{manfnt}
\usepackage{graphicx}
\usepackage{color}
\usepackage{multirow}
\usepackage{tabularx,colortbl}
\usepackage{alltt}
\usepackage{bussproofs}
\usepackage[linesnumbered,vlined]{algorithm2e}
\usepackage{pifont}
\usepackage{cancel}
\usepackage{rotating}

\usepackage{todonotes}
\usepackage{latexsym,amssymb,amsmath,amsfonts}

  { }

\newcommand{\diff}{\ensuremath{\mathtt{diff}}}

\newcommand{\ARRAY}{\ensuremath{\mathtt{ARRAY}}\xspace}
\newcommand{\INDEX}{\ensuremath{\mathtt{INDEX}}\xspace}
\newcommand{\ELEM}{\ensuremath{\mathtt{ELEM}}\xspace}
\newcommand{\formulae}{formul\ae\xspace}

\newcommand{\LRA}{\ensuremath{\mathcal{LRA}}\xspace}
\newcommand{\EUF}{\ensuremath{\mathcal{EUF}}\xspace}
\newcommand{\LIA}{\ensuremath{\mathcal{LIA}}\xspace}
\newcommand{\IDL}{\ensuremath{\mathcal{IDL}}\xspace}

\newcommand{\AXEXT}{\ensuremath{\mathcal{AR}_{{\rm ext}}}\xspace}

\newcommand{\imp}{\rightarrow}

\newcommand{\dpll}[1]{{\sc DPLL}\xspace}

\newcommand{\COMMENT}[1]{}

\newcommand{\ua}{\ensuremath{\underline a}}

\newcommand{\ue}{\ensuremath{\underline e}}

\newcommand{\ux}{\ensuremath{\underline x}}

\newcommand{\uy}{\ensuremath{\underline y}}
\newcommand{\uz}{\ensuremath{\underline z}}

\newcommand{\cA}{\ensuremath \mathcal A}
\newcommand{\cB}{\ensuremath \mathcal B}
\newcommand{\cM}{\ensuremath \mathcal M}
\newcommand{\cN}{\ensuremath \mathcal N}

\newcommand{\cC}{\ensuremath \mathcal C}

\newcommand{\cI}{\ensuremath \mathcal I}

\renewcommand{\int}{\ensuremath {\mathcal I}}

\newcommand{\AXD}{\ensuremath{\mathcal{ARD}}\xspace}
\newcommand{\AXDTI}{\ensuremath{\mathcal{\AXD}(T_I)}\xspace}
\newcommand{\AXEXTTI}{\ensuremath{\mathcal{\AXEXT}(T_I)}\xspace}

\newcommand{\uguale}{\ensuremath{=}\xspace}
\newcommand{\coincide}{\ensuremath{\equiv}\xspace}

\newcommand{\eop}{$\hfill \dashv$}

\usepackage[]{algorithm2e}
\usepackage{framed}
\usepackage{tabularx}
\usepackage{tikz-cd}
\usepackage{wrapfig}
\usepackage{comment}
\usepackage{xspace}
\usepackage{textgreek}
\usepackage{paralist}
\usepackage{enumerate}

\begin{document}
\title{Interpolation and Amalgamation \\ for Arrays with MaxDiff (Extended Version)}
\titlerunning{Interpolation and Amalgamation for Arrays with MaxDiff}
%
\author{Silvio Ghilardi\inst{1}
\and
Alessandro Gianola\inst{2}
\and
Deepak Kapur\inst{3}
}
\authorrunning{S. Ghilardi et al.}
%
\institute{%
 Dipartimento di Matematica, Universit\`a degli Studi di Milano (Italy)\\
\and 
Faculty of Computer Science, Free University of Bozen-Bolzano (Italy)\\
\email{gianola@inf.unibz.it}
\and
Department of Computer Science, University of New Mexico (USA)\\
}
\maketitle              
\begin{abstract}
In this paper, the theory of McCarthy's extensional arrays enriched with a maxdiff operation (this operation returns the biggest index where two given arrays differ) is proposed. It is known from the literature that a diff operation is required  for the theory of arrays in order to enjoy 
the Craig interpolation property at the quantifier-free level. However, the diff operation introduced in the literature is merely instrumental to this purpose and has only a purely formal meaning (it is obtained from the Skolemization of the extensionality axiom). Our maxdiff operation significantly increases the level of expressivity; 
however, obtaining interpolation results for the resulting theory becomes a surprisingly hard task. We obtain such results via a thorough semantic analysis of the models of the theory and of their amalgamation properties. The results are modular with respect to the index theory and 
it is shown how to convert them into concrete interpolation algorithms via a hierarchical approach.

\end{abstract}

\section{Introduction}\label{sec:intro}

Since McMillan's seminal papers~\cite{McM03,McM06}, interpolation 
 has been successfully applied in software model checking, also in combination with orthogonal techniques like PDR~\cite{VG14} or $k$-induction~\cite{KVGG19}.
The reason why interpolation techniques are so attractive is because they allow to discover in a completely \emph{automatic} way new 
atoms 
(improperly often called `predicates') 
that might contribute to the construction of invariants. In fact, software model-checking problems are typically infinite state, so invariant synthesis may require introducing formulae 
whose search is not finitely bounded. One way to discover them is to analyze spurious error traces; for instance, 
if the system under examination (described by a transition formula $Tr(\ux, \ux')$)  cannot reach in $n$-step an error configuration in $U(\ux)$
starting from an initial configuration in 
$In(\ux)$, this means that the formula
$$
In(\ux_0)\wedge Tr(\ux_0, \ux_1)\wedge \cdots \wedge Tr(\ux_{n-1}, \ux_n) \wedge U(\ux_n) 
$$
is inconsistent (modulo a suitable theory $T$). From the inconsistency proof, 
 by computing an interpolant, say at the $i$-th iteration, one can produce a formula $\phi(\ux)$  
such that, modulo $T$, we have
\begin{equation}\label{eq:trace}
In(\ux_0)\wedge \bigwedge_{j=0}^i Tr(\ux_{j-1}, \ux_j)\models \phi(\ux_i)
~~{\rm and}~~\phi(\ux_i)\wedge \bigwedge_{j=i+1}^n Tr(\ux_{j-1}, \ux_j)
\wedge U(\ux_n) \models \bot.
\end{equation}
This formula (and the atoms it contains) can contribute to the refinement of the current candidate 
loop invariant guaranteeing safey.  
This fact can be exploited in very different ways during invariant search, depending on the various techniques employed.
It should be noticed however that interpolants are not unique and that different interpolation algorithms may return interpolants of different quality: all interpolants restrict search, but not all of them might be conclusive.

 This new application of interpolation is different from the role of interpolants for analyzing proof theories of various logics starting with the pioneering works 
 of \cite{Craig,HuangInterpolant,Pudlak97}. It should be said however that
Craig interpolation theorem in first order
logic  does not give by itself any information on the shape the interpolant can have when a specific theory is involved. Nevertheless, 
this is crucial for the applications: when we extract an interpolant from a trace like~\eqref{eq:trace}, we are typically handling a theory which might be undecidable, but whose quantifier-free fragment  is decidable for satisfiability (usually within a somewhat `reasonable' computational complexity). Thus, 
it is desirable (although not always possible) that  the interpolant is quantifier-free, a fact which is not guaranteed  in the general case. This is why a lot of effort has been made in analyzing \emph{quantifier-free} interpolation, also exploiting its connection to semantic properties like amalgamation and strong amalgamation (see~\cite{BGR14} for comprehensive results in the area).

The specific theories we want to analyze in this paper are variants of \emph{McCarthy's theory of arrays}~\cite{mccarthy} \emph{with extensionality} (see Section~\ref{sec:tharr} below for a detailed description). 
The main operations considered in this theory are the \emph{write} operation 
(i.e. the array update) and the \emph{read} operation (i.e., the access to the content of an array cell). As such, this theory is suitable to formalize programs over arrays, like standard copying, comparing, searching, sorting,  etc. functions; verification problems of this kind are collected in the SV-COMP benchmarks category ``ReachSafety-Arrays''\footnote{
\url{https://sv-comp.sosy-lab.org/2020/benchmarks.php}
}, where  safety verification tasks involving arrays of \emph{finite but unknown length} are considered.
%
%

By itself, the theory of arrays with extensionality does not have quantifier free interpolation~\cite{KMZ06}\footnote{ This is the counterexample (due to R. Jhala):
the formula $x=wr(y,i,e) $
is inconsistent with the formula $rd(x,j)\neq rd(y,j)\wedge rd(x,k)\neq rd(y,k)\wedge j\neq k$, but all possible interpolants require quantifiers to be written 
(with diff symbols, instead, it is possible to write down an interpolant without quantifiers, as shown in~\cite{axdiff}). 
}; 
 however, in~\cite{axdiff} it was shown that quantifier-free interpolation is restored if one enriches  the language with a binary function  skolemizing the extensionality axiom (the result was confirmed - via different interpolation algorithms - in~\cite{HS18,TW16}).  Such a Skolem function, applied to two array variables $a,b$, returns an index $\diff(a,b)$ where $a,b$ differ (it returns an arbitrary value if $a$ is equal to $b$). This semantics for the $\diff$ operation is very undetermined and does not have a significant interpretation in concrete programs. 
That is why we propose to modify it in order 
to give it a defined and natural meaning: 
we ask for $\diff(a,b)$ to return the \emph{biggest} index where $a,b $ differ (in case $a=b$ we ask for $\diff(a,b)$ to be the minimum index $0$). Since it is  natural to view arrays as functions defined on initial intervals of the nonnegative integers, this choice has a clear semantic motivation. The expressive power of the theory of arrays so enriched becomes bigger: for instance, if we also add to the language a constant symbol $\epsilon$ for the 
 undefined array constantly equal to some `undefined' value $\bot$ 
(where $\bot$ is meant to be different from the values $a[i]$ actually in use), then we can define
$\vert a\vert$ as $\diff(a, \epsilon)$.
 In this way  we can model the fact that $a$ is undefined outside the interval $[\, 0,\vert a\vert\,]$- this is 
useful 
to formalize the above mentioned SV-COMP benchmarks.

The effectiveness of quantifier-free interpolation in the theory of arrays with maxdiff is exemplified in the simple example  of Figure~\ref{Fig1}: the invariant certifying the assert in line 7 of the  \texttt{Strcpy} algorithm can be obtained taking a
suitable quantifier-free interpolant  out of the spurious trace~\eqref{eq:trace} already for $n=2$.
%
In more realistic examples, as witnessed by current research~\cite{ABGRS12a,ABGRS12,ABGRS14,AGS14,FPMG19,GurfinkelSV18,IIRS20,CGU20}, it is quite clear that useful invariants require universal quantifiers to be expressed and if  undecidable fragments are invaded, incomplete solvers must be used. However, even in such circumstances, quantifier-free interpolation does not 
lose its interest: for instance, the tool \textsc{Booster}~\cite{AGS14}\footnote{ \textsc{Booster} is no longer maintained, however it is still 
referred to in
current experimental evaluations~\cite{FPMG19,CGU20}.
} synthesizes universally quantified invariants out of quantifer-free interpolants (quantifier-free interpolation problems are generated by negating and skolemizing universally quantified formulae arising during invariants  search, see~\cite{ABGRS14} for details).

\begin{figure}
\begin{center}
\begin{tabular}{|l|c|r|}
\begin{minipage}{.32\textwidth}
\begin{algorithm}[H]
$\mathtt{int~ a[N];}$
 \\ 
$\mathtt{ int ~b[N];}$\\ 
$\mathtt{int ~I=0;}$ \\
 \While{$\mathtt{I<N}$}{
 $\mathtt{b[I]= a[I];}$
  \\
  $\mathtt{I++;}$
  }
  $\mathtt{assert(a=b);}$
 \end{algorithm} 
\end{minipage}
\begin{minipage}{.65\textwidth}
 \begin{itemize}
  \item $In(a,b,I)~\equiv~I=0\wedge \vert a \vert=N-1 \wedge \vert b \vert=N-1 \wedge N>0~$
  \item $Tr(a,b,I, a', b', I')~\equiv~ I<N \wedge I'=I+1\wedge a'=a~\wedge ~b'=wr(b,I, rd(a,I))$
  \item $U(a,b)~\equiv~a\neq b\wedge I=N$
 \end{itemize}
\end{minipage}
 
\end{tabular}
\end{center}
 \caption{\texttt{Strcpy} function: code and associated transition
 system (with program counter  missed in the latter for simplicity).
 \newline
 Loop invariant: $a=b \vee (N > \diff(a,b) \wedge \diff(a,b) \geq I)$.  
 }\label{Fig1}
\end{figure}

Proving that the theory of arrays with the above `maxdiff' operation enjoys quantifier-free interpolation revealed to be a surprisingly difficult task. In the end, the interpolation algorithm we obtain resembles   the interpolation algorithms generated via the hierarchic locality  techniques introduced in~\cite{SS08,SS18} and employed also in~\cite{TW16}; however, its correctness, completeness and termination proofs require a large d\'etour going through non-trivial model-theoretic arguments (these arguments do not substantially simplify adopting the complex framework of `amalgamation closures' and `$W$-separability' of~\cite{TW16}, and that is the reason why we preferred to supply direct proofs). 

This paper concentrates on theoretical and methodological results, rather than on experimental aspects.
It is almost completely dedicated to the correctness and completeness poof of our interpolation algorithm: in Subsection~\ref{subsec:road} we summarize our proof plan and supply basic intuitions.
The paper is structured as follows: in Section~\ref{sec:background} we recall some background, in Section~\ref{sec:tharr} we introduce our theory of arrays with maxdiff; Sections~\ref{sec:embeddings} and~\ref{sec:semantic-arg} supply the semantic proof of the amalgamation theorem; Sections~\ref{sec:sat} and~\ref{sec:interpolation} are dedicated to the algorithmic aspects, whereas Section~\ref{sec:to} analyzes complexity 
 for the restricted case where indexes are constrained by the theory of total orders. In the final Section~\ref{sec:conclusions}, we mention some still open problems. 
  The main results in the paper are Theorems~\ref{thm:axd_amalg},\ref{thm:al},\ref{thm:allin}: for space reasons, \emph{all proofs of these theorems will be only sketched}, full details are nevertheless supplied in the Appendix~\ref{sec:app}. Appendix~\ref{sec:app} contains also additional material on complexity analysis and implementation. It contains also a proof about nonexistence of uniform interpolants (see~\cite{kapur,kapurJSSC,CILC20,cade19,IJCAR20,JAR21} for the definition and more information on uniform interpolants). 



\section{Formal Preliminaries}
\label{sec:background}

We assume the usual syntactic (e.g., signature, variable, term, atom,
literal, formula, and sentence) and semantic (e.g., structure,
sub-structure, truth, satisfiability, and validity) notions of
(possibly many-sorted) first-order logic.
The equality symbol
``\uguale'' is included in all signatures considered below. 
Notations like $E(\ux)$ mean that the expression (term, literal,
formula, etc.) $E$ contains free variables only from the tuple $\ux$.
A `tuple of variables' is a list of variables without repetitions and
a `tuple of terms' is a list of terms (possibly with repetitions).
Finally, whenever we use a notation like $E(\ux, \uy)$ we implicitly
assume not only that both the $\ux$ and the $\uy$ are pairwise
distinct, but also that $\ux$ and $\uy$ are disjoint. 
A \emph{constraint} is a conjunction of literals.
A formula is
\emph{universal} (\emph{existential}) iff it is obtained from a
quantifier-free formula by prefixing it with a string of universal
(existential, resp.)  quantifiers.

\paragraph{Theories and satisfiability modulo theory.}
 A \emph{theory} $T$ is a pair $({\Sigma},
Ax_T)$, where $\Sigma$ is a signature and $Ax_T$ is a set of
$\Sigma$-sentences, called the \emph{axioms} of $T$ (we shall
sometimes write directly $T$ for $Ax_T$).  The \emph{models} of $T$
are those $\Sigma$-structures in which all the sentences in $Ax_T$ are
true.  
A $\Sigma$-formula $\phi$ is \emph{$T$-satisfiable} 
(or $T$-consistent)
if there exists a
model $\cM$ of $T$ such that $\phi$ is true in $\cM$ under a suitable
assignment $\mathtt a$ to the free variables of $\phi$ (in symbols,
$(\cM, \mathtt a) \models \phi$); it is \emph{$T$-valid} (in symbols,
$T\vdash \varphi$) if its negation is $T$-unsatisfiable or,
equivalently, $\varphi$ is provable from the axioms of $T$ in a
complete calculus for first-order logic.  
A theory $T=(\Sigma, Ax_T)$ is \emph{universal} iff 
all sentences in $Ax_{T}$ are
universal.
A formula $\varphi_1$ \emph{$T$-entails} a formula $\varphi_2$ if
$\varphi_1 \to \varphi_2$ is \emph{$T$-valid} (in symbols,
$\varphi_1\vdash_T \varphi_2$ or simply $\varphi_1 \vdash \varphi_2$
when $T$ is clear from the context). 
If $\Gamma$ is a set of \formulae and $\phi$ a formula, $\Gamma \vdash_T \phi$ means that there are 
$\gamma_1, \dots, \gamma_n\in \Gamma$ such that $\gamma_1\wedge \cdots\wedge \gamma_n \vdash_T \phi$.  
 The \emph{satisfiability modulo
  the theory $T$} (SMT$(T)$) \emph{problem} amounts to establishing
the $T$-satisfiability of quantifier-free $\Sigma$-\formulae (equivalently, 
the $T$-satisfiability of  $\Sigma$-constraints).
A theory $T$ admits \emph{quantifier-elimination} iff for every
formula $\phi(\ux)$ there is a quantifier-free formula
$\phi'(\ux)$ such that $T\vdash \phi \leftrightarrow \phi'$. 

Some theories have special names, 
which are becoming standard in SMT-literature; 
for instance, $\EUF(\Sigma)$ is the pure equality theory in the signature $\Sigma$ (this is commonly abbreviated as \EUF if
there is no need to specify the signature $\Sigma$).
More standard theory names will be recalled during the paper.
%

\paragraph{Embeddings and sub-structures} 

 The support 
 of a structure $\mathcal{M}$ is denoted
with $|\mathcal{M}|$. For a (sort, function, relation) symbol $\sigma$, we denote as 
$\sigma^\cM$ the interpretation of $\sigma$ in $\cM$.
An embedding is a homomorphism that preserves
and reflects relations and operations (see, e.g.,~\cite{CK}).
Formally, a {\it $\Sigma$-embedding} (or, simply, an embedding)
between two $\Sigma$-structu\-res $\cM$ and $\cN$ is any mapping $\mu:
|\cM| \longrightarrow |\cN|$ satisfying the following three
conditions: (a) it is a
(sort-preserving) 
injective function; (b) it is an algebraic homomorphism, that is for
every $n$-ary function symbol $f$ and for every $a_1, \dots, a_n\in
|\mathcal{M}|$, we have $f^{\cN}(\mu(a_1), \dots, \mu(a_n))=
\mu(f^{\cM}(a_1, \dots, a_n))$; (c) it preserves and reflects
predicates, i.e.\ for every $n$-ary predicate symbol $P$,
we have $(a_1, \dots, a_n)\in P^{\cM}$ iff $(\mu(a_1), \dots,
\mu(a_n))\in P^{\cN}$.  If $|\mathcal{M}|\subseteq |\mathcal{N}|$ and
the embedding $\mu: \cM \longrightarrow \cN$ is just the identity
inclusion $|\mathcal{M}|\subseteq |\mathcal{N}|$, we say that $\cM$ is
a {\it substructure} of $\cN$ or that $\cN$ is a {\it superstructure}
of $\cM$.
As it is known, 
 the truth of a universal (resp. existential)
sentence is preserved through substructures (resp. superstructures).

\paragraph{Combinations of
    theories.}
A theory $T$ is \emph{stably infinite} iff every $T$-satisfiable
quantifier-free formula (from the signature of $T$) is satisfiable in
an infinite model of $T$.  By compactness, it is possible to show that
$T$ is {stably infinite} iff every model of $T$ embeds into an
infinite one (see, e.g., \cite{Ghil05}).  A theory $T$ is \emph{convex} iff
for every conjunction of literals $\delta$, if $\delta\vdash_T
\bigvee_{i=1}^n x_i=y_i$ then $\delta\vdash_T x_i=y_i$ holds for some
$i\in \{1, ..., n\}$.
Let $T_i$ be a stably-infinite theory over the signature $\Sigma_i$
such that the $SMT(T_i)$ problem is decidable for $i=1,2$ and such that 
$\Sigma_1$ and $\Sigma_2$ are disjoint (i.e.\ the only shared symbol
is equality).  Under these assumptions, the \emph{Nelson-Oppen combination
result}~\cite{NO79} 
says 
that the SMT problem for the combination
$T_1\cup T_2$ of the theories $T_1$ and $T_2$ 
 is
decidable.   

\paragraph{Interpolation properties.}
%
Craig's interpolation theorem~\cite{CK} roughly states that if a
formula $\phi$ implies a formula $\psi$ then there is a third formula
$\theta$, called an interpolant, such that $\phi$ implies $\theta$, 
$\theta$ implies $\psi$, and every non-logical symbol in $\theta$
occurs both in $\phi$ and $\psi$.  Our interest is to
specialize this result to the computation of quantifier-free
interpolants modulo (combinations of) theories. 

\begin{definition}\label{def:restricted}[Plain quantifier-free interpolation]
  A theory $T$ \emph{admits (plain) quantifier-free interpolation}
  (or, equivalently, \emph{has quantifier-free interpolants}) iff for
  every pair of quantifier-free formulae $\phi, \psi$ such that
  $\psi\wedge \phi$ is $T$-unsatisfiable, there exists a
  quantifier-free formula $\theta$, called an \emph{interpolant}, such
  that: (i) $\psi$ $T$-entails $\theta$, (ii) $\theta \wedge \phi$ is
  $T$-unsatisfiable, and (iii) only the variables occurring in both
  $\psi$ and $\phi$ occur in $\theta$.
\end{definition}
In verification, the following extension of Definition \ref{def:restricted} is
considered more useful. 

\begin{definition}
  \label{def:general}[General quantifier-free interpolation]
  Let $T$ be a theory in a signature $\Sigma$; we say that $T$ has the
  \emph{general quantifier-free interpolation property} iff for every
  signature $\Sigma'$ (disjoint from $\Sigma$) and for every pair of
  ground $\Sigma\cup\Sigma'$-\formulae $\phi, \psi$ such that
  $\phi\wedge \psi$ is $T$-unsatisfiable\footnote{By this (and
    similar notions) we mean that $\phi\wedge \psi$ is unsatisfiable
    in all $\Sigma'$-structures whose $\Sigma$-reduct is a model of
    $T$. },  there is a ground formula $\theta$ such that: (i) $\phi$
  $T$-entails $\theta$; (ii) $\theta\wedge \psi$ is $T$-unsatisfiable;
  (iv) all 
  relations, constants and function symbols from $\Sigma'$
  occurring in $\theta$ also occur in $\phi$ and $\psi$.
\end{definition}
By replacing free variables with free constants, it should be clear
that general quantifier-free interpolation
(Definition~\ref{def:general}) implies plain quantifier-free
interpolation (Definition~\ref{def:restricted}); however, the converse implication does not hold. 

\paragraph{Amalgamation and strong amalgamation.}

Interpolation can be characterized semantically via amalgamation.

\begin{definition}
  \label{def:sub-amalgamation-and-strong-sub-amalgamation}
  A universal theory $T$ has the \emph{amalgamation property} iff 
  given models $\cM_1$ and $\cM_2$ of $T$ and a common
  submodel $\cA$ of them, there exists a further model $\cM$ of
  $T$ (called $T$-amalgam) 
  endowed with embeddings $\mu_1:\cM_1 \longrightarrow \cM$ and
  $\mu_2:\cM_2 \longrightarrow \cM$ whose restrictions to $|\cA|$
  coincide.
    %

  A universal theory $T$ has the \emph{strong amalgamation property} if the
  above embeddings $\mu_1, \mu_2$ and the above model $\cM$ can be chosen so to satisfy the following additional
  condition: if, for some $m_1\in \vert \cM_1\vert , m_2\in \vert \cM_2\vert$, 
  $\mu_1(m_1)=\mu_2(m_2)$ holds,
  then there exists an element $a$ in $|\cA|$ such that $m_1=a=m_2$.
\end{definition}

The first statement  of the following theorem is an old result due to~\cite{amalgam}; the second statement is proved in~\cite{BGR14} (where it is also suitably reformulated for theories which are not universal):

\begin{theorem}
  \label{thm:interpolation-amalgamation} Let $T$ be a
universal theory. Then 
\begin{compactenum}
 \item[{\rm (i)}] $T$ has the amalgamation property iff it admits
  quan\-ti\-fier-free interpolants;
  \item[{\rm (ii)}] $T$ has the strong amalgamation property iff it has the general quantifier-free interpolation property.
\end{compactenum} 
\end{theorem}

We underline that, in presence of stable infiniteness, strong amalgamation is a modular property (in the sense that it transfers to signature-disjoint unions of theories), whereas amalgamation is not (see again~\cite{BGR14} for details).

\section{Arrays with MaxDiff}
\label{sec:tharr}

The  \emph{McCarthy theory of arrays}~\cite{mccarthy} has three
sorts $\ARRAY, \ELEM, \INDEX$ (called ``array'', ``element'', and
``index'' sort, respectively) and two function symbols $rd$ (``read'') and $wr$ (``write'')  
of appropriate arities; its axioms are:
\begin{eqnarray*}
  \forall y, i, e. & & rd(wr(y,i,e),i) \uguale e \\
  \forall y, i, j, e. & & i \not\uguale j \imp rd(wr(y,i,e),j)\uguale rd(y,j) .
\end{eqnarray*}
The McCarthy theory of \emph{arrays with extensionality} 
has the further
axiom
\begin{eqnarray}\label{ext}
 \forall x, y. 
 x \not\uguale y \imp (\exists i.\ rd(x,i)\not\uguale rd(y,i)),
\end{eqnarray}
called the `extensionality' axiom. 
The theory of arrays with extensionality is not universal and quantifier-free interpolation fails for it~\cite{KMZ06}. 
In~\cite{axdiff}  a
variant of the McCarthy {theory of arrays} with extensionality,
obtained by Skolemizing the axioms of extensionality, is introduced. This variant of the theory turns out to be universal and to enjoy quantifier-free interpolation.
However, the Skolem function introduced in~\cite{axdiff} is generic, here we want to make it more informative, so as to return the biggest index where two different arrays differ. 
To locate our contribution in the 
general context, we need the notion of an index theory.  

\begin{definition}\label{def:index}
 An \emph{index theory} $T_I$ is a mono-sorted theory (let \INDEX be its sort) satisfying the following conditions:
 \begin{compactenum}
  \item[-] $T_I$ is universal, stably infinite and 
  has the general quantifier-free interpolation property (i.e. it is strongly amalgamable, see Theorem~\ref{thm:interpolation-amalgamation});
  \item[-] $SMT(T_I)$ is decidable;
  \item[-] $T_I$ extends the theory $TO$ of linear orderings with a distinguished  element $0$.
 \end{compactenum}
\end{definition}
We recall that $TO$ is the theory whose only proper symbols (beside equality) are a
 binary predicate $\leq$ and a constant $0$ subject to the  axioms saying that $\leq$ is reflexive, transitive, antisymmetric and total (the latter means that $ i\leq j \lor j\leq i $ holds for all $i,j$).
 Thus, the signature of an index theory $T_I$ contains at least the binary relation symbol $\leq$
  and the constant $0$. In the paper, 
 by  a $T_I$-term, $T_I$-atom, $T_I$-formula, etc. we mean a term, atom, formula in the signature of $T_I$.
  %
 Below, we use the abbreviation $i<j$ for $i\leq j \land i\not \uguale j$. The constant $0$ is meant to 
 separate `formally positive' indexes - those satisfying $0\leq i$ - from the remaining `formally negative' ones. 

Examples of index theories are $TO$ itself, integer difference logic \IDL, integer linear arithmetic \LIA, and  real linear arithmetics \LRA. In order to match the  requirements of Definition~\ref{def:index}, one must however make a careful choice of the language, see~\cite{BGR14} for details: the most important detail is that integer (resp. real) division  by all positive 
integers should be added to the language of \LIA (resp. \LRA). For most applications, \IDL 
(namely the theory of integer numbers with 0, ordering, successor and predecessor)~\footnote{ 
The name 'integer difference logic' comes from the fact that atoms in this theory are equivalent to \formulae of the kind $S^n(i)\Join j$ (where $\Join\,\in \{\leq, \geq, =\}$), thus they represent difference bound constraints of the kind  $j-i\Join n$ for $n\geq 0$.
} suffices as in this theory one can model counters for scanning arrays.

Given an index theory $T_I$, we now introduce our \emph{array theory with maxdiff} $\AXD(T_I)$ (parameterized by $T_I$) as follows. We still have three sorts $\ARRAY, \ELEM, \INDEX$; the language includes the symbols of $T_I$, the read and write operations $rd, wr$, 
a binary function $\diff$
of type $\ARRAY \times \ARRAY \to \INDEX$, 
as well as constants $\epsilon$ and $\bot$ of sorts \ARRAY and \ELEM, respectively. 
The constant $\bot$ models 
an undetermined (e.g. undefined, not-in-use, not coming from appropriate initialization, etc.) value
and $\varepsilon$ models the totally undefined array; the term $\diff(x,y)$ returns the maximum index where $x$ and $y$ differ and returns 0 if $x$ and $y$ are equal.~\footnote{Notice that it might well be the case that 
$\diff(x,y)=0$ for different $x,y$, but in that case $0$ is the only index where 
$x,y$ differ. }
Formally, the axioms of $\AXD(T_I)$
 include, besides the axioms of $T_I$, the following ones:
\begin{eqnarray}
  \label{ax1}
  \forall y, i, e. & & i\geq 0 \to rd(wr(y,i,e),i) \uguale e \\
  \label{ax2}
  \forall y, i, j, e. & & i \not\uguale j \imp rd(wr(y,i,e),j)\uguale rd(y,j)\\
  \label{ax3}
  \forall x, y. & & x \not\uguale y \imp rd(x,\diff(x,y))\not\uguale rd(y,\diff(x,y)) \\
  \label{ax4}
  \forall x, y, i. & & i> \diff(x,y) \imp rd(x,i) = rd (y,i)\\
  \label{ax5}
  \forall x. & & \diff(x,x)=0 \\
  \label{ax6}
  \forall x.i & & i<0 \imp rd(x,i)=\bot \\
  \label{ax7}
  \forall i. & & rd(\varepsilon,i)=\bot
\end{eqnarray} 
In the read-over-write axiom~\eqref{ax1}, we put the proviso $i\geq 0$ because we want all our arrays to be undefined on negative indexes (negative updates makes no sense and  have no effect: by axiom ~\eqref{ax6}, reading a negative index always produces $\bot$).

We call \AXEXTTI (the `theory of arrays with extensionality parameterized by $T_I$') the theory
obtained from \AXDTI by removing the symbol $\diff$ and by replacing the axioms~\eqref{ax3}-\eqref{ax5} by the extensionality axiom~\eqref{ext}.
Since the extensionality axioms follows from axiom~\eqref{ax3},  $\AXDTI$ is  an extension of $\AXEXTTI$.

As an effect of the above axioms, we have that an array $x$ is undefined outside the interval 
$[0,\vert x\vert]$, where $\vert x\vert$
is defined as $\vert x\vert := \diff(x, \varepsilon)$. Typically, this interval is finite and in fact our proof of Theorem ~\ref{thm:sat} below  shows that any satisfiable constraint is satisfiable in a model where all such intervals (relatively to the variables involved in the constraint) are finite.

The 
next lemma is immediate from the 
axiomatization of \AXDTI: 

\begin{lemma}\label{lem:univinst}
 An atom of the form 
 $a=b$ is equivalent (modulo \AXDTI) to 
 \begin{equation}\label{eq:eleq}
  \diff(a,b)=0 \wedge rd(a,0)=rd(b,0)~.
 \end{equation}
An atom of the form 
$a=wr(b,i,e)$ is equivalent (modulo \AXD) to
\begin{equation}\label{eq:elwr}
(i\geq 0 \to rd(a,i)=e) ~\wedge~ \forall h~(h\neq i \to rd(a,h)=rd(b,h))~.
\end{equation}
An atom of the form 
$\diff(a,b)=i$ is equivalent (modulo \AXDTI) to
\begin{equation}\label{eq:eldiff}
i\geq 0 ~\wedge ~\forall h~(h>i \to rd(a,h)=rd(b,h)) ~\wedge ~(i> 0\to rd(a,i)\neq rd(b,i))~~.  
\end{equation}
\end{lemma}

For our interpolation algorithm in Section~\ref{sec:interpolation},
we need to introduce iterated $\diff$ operations, similarly to~\cite{TW16}. As we know $\diff(a,b)$ returns the biggest index where $a$ and $b$ differ (it returns 0 if $a=b$). Now we want an operator that returns the last-but-one index where $a,b$ differ (0 if $a,b$ differ in at most one index), an operator that returns the last-but-two index where $a,b$ differ (0 is they differ in at most two indexes), etc. Our language is already enough expressive for that, so we can introduce such operators explicitly  as follows. Given array variables $a,b$, we define by mutual recursion the sequence of array terms $b_1, b_2, \dots$ and of index terms $\diff_1(a,b), \diff_2(a,b), \dots$:
\begin{eqnarray*}
b_1~:=~b;~~~~~~~~~~~~~~~~~~~~~~~~~ ~~~~~~~~~~~~~~~~~~~~~~~~\;
&& \diff_1(a,b)~~~~:=~  \diff(a,b_1);~~~~~~~~~~~~ \\
b_{k+1}~:=~wr(b_k, \diff_k(a,b), rd(a,\diff_k(a,b))); && \diff_{k+1}(a,b)~:=~\diff(a, b_{k+1}) 
\end{eqnarray*}

\noindent
Intuitively, $b_{k+1}$ is the same as $b$ except for all $k$-last indexes on which $a$ and $b$ differ, in correspondence of which $b_{k+1}$ has the same value as $a$. A useful fact is that conjunctions of formulae of the kind $\bigwedge_{j<l}\diff_j(a,b)=k_j$ can be eliminated in favor of universal clauses in a language whose only symbol  for array variables is $rd$. In detail:

\begin{lemma}\label{lem:elim}
 A formula like 
 \begin{equation}\label{eq:iterated_diff}
  \diff_1(a,b)=k_1 \wedge \cdots \cdots \wedge \diff_l(a,b)=k_l
 \end{equation}
is equivalent modulo $\AXD$ to the conjunction of the following five formulae:
\begin{eqnarray}
 & k_1\geq k_2 \wedge \cdots \wedge k_{l-1}\geq k_l\wedge k_l\geq 0 \label{s0} \\
 & \bigwedge_{j<l} (k_j> k_{j+1} \to rd(a,k_j)\neq rd(b,k_j)) \label{s1}
   \\
   &
     \bigwedge_{j<l} (k_j= k_{j+1} \to k_j=0) \label{s2}
     \\
   & \bigwedge_{j\leq l} (rd(a,k_j)=rd(b,k_j)\to k_j=0) \label{s3}
   \\
   &
   \forall h~(h>k_l \to rd(a,h)=rd(b,h) \vee h=k_1\vee\cdots\vee h=k_{l-1})  
   \label{s4}
\end{eqnarray}
\end{lemma}

\subsection{Our roadmap}\label{subsec:road}

The main result of the paper is that, for every index theory $T_I$, the array theory with maxdiff \AXDTI indexed by $T_I$ \emph{enjoys quantifier-free interpolation} and that \emph{interpolants can be computed hierarchically} by relying on a black-box quantifier-free interpolation algorithm for the weaker theory $T_I\cup \EUF$ (the latter theory has quantifier free interpolation because $T_I$ is strongly amalgamable and because of Theorem~\ref{thm:interpolation-amalgamation}).
In this subsection, we supply intuitions and we give a qualitative high-level view to our proofs: more technical details and full proofs can be found in Appendix \ref{sec:app}.  

\subsection*{The algorithm.} By general easy transformations (recalled in Section~\ref{sec:interpolation} below), it is sufficient to be able to extract a quantifier-free interpolant out of a pair of quantifier-free formulae $A, B$ such that (i) $A\wedge B$ is \AXDTI-inconsistent; (ii) both $A$ and $B$ are conjunctions of flat literals, i.e. of literals which are equalities between variables, disequalities between variables or literals of the form 
$R(\ux), \neg R(\ux), f(\ux)=y$ (where $\ux,y$ are variables, $R$ is a predicate symbol and $f$ a function symbol). 

Let us call \emph{common} the variables occurring in both $A$ and $B$.
The fact that a quantifier-free interpolant exists intuitively means that there are two reasoners (an $A$-reasoner operating  on formulae involving only the variables occurring in $A$ and a $B$-reasoner operating  on formulae involving only the variables occurring in $B$) that are able to discover the inconsistency of $A\land B$ by exchanging information on the common language, i.e. by communicating each other only the entailed quantifier-free formulae involving  the common variables.

A  problem that can be addressed  when designing an interpolation algorithm, is that 
there are infinitely many common terms that can be built up out  of finitely many common variables and it may happen that some uncommon terms can be recognized to be
equal to some common terms during the deductions performed by the $A$-reasoner and the $B$-reasoner.

As an example, suppose that $A$ contains 
the literals $c_1=wr(c_2, i, e), c_1\neq c_2, a=wr(c_3, i, e)$, where only $c_1, c_2, c_3$ are 
common (i.e. only these variables occur in $B$). Then using diff operations, we can deduce $i=\diff(c_1, c_2), e=rd(c_1, i)$ so that in the end we can conclude that $a$ is also `common', being definable in term of common variables. Thus, the $A$-reasoner must communicate 
(via a defining common term or in some other indirect way)
to the $B$-reasoner any fact it discovers about $a$, although $a$ was not listed among the common variables since the very beginning. In more sophisticated examples, iterated diff operations are needed  to discover `hidden' common facts. 

To cope with the above problem,  our algorithm \emph{gives names} $i_k= \diff_k(c_1,c_2)$ to all the iterated diffs 
of common array variables $c_1, c_2$ (the newly introduced names $i_k$ are considered common and can be replaced back with their defining terms when  the interpolants are computed at the end of the algorithm).

The second component of our algorithm is \emph{instantiation}. Both the $A$- and the $B$-reasoner use the content of Lemmas~\ref{lem:univinst} and~\ref{lem:elim} in order to handle atoms of the kind $a=b$, $a_1=wr(a_2, i,e)$, $i=\diff_k(a_1, a_2)$. Whenever they come across such atoms, the equivalent \formulae supplied by these lemmas are 
taken into consideration; in fact, whenever 
the lemmas produce universally quantified clauses of the kind $\forall h\, C$, 
they replace in $C$ the universally quantified index variable $h$ by \emph{all possible instantiations} with their own index terms  (these are the terms built up from index variables occurring in $A$ for the $A$-reasoner and  occurring in $B$ for the $B$-reasoner  respectively). 
 Such instantiations can be read as \emph{clauses in the language of 
$T_I\cup \EUF$} if we replace every array variable $a$ by a fresh unary function symbol $f_a$ and read terms like $rd(a,i)$ as $f_a(i)$.

Of course both the production of names for iterated diff-terms and the instantiation with owned index terms 
need to be repeated 
(possibly, infinitely many times);  
we prove however (this is the content of our main Theorem~\ref{thm:al} below) that \emph{if $A\land B$ is \AXDTI-inconsistent, then sooner or later the union of the sets of the  clauses deduced by the $A$-reasoner and the $B$-reasoner in the restricted signature of $T_I\cup \EUF$ is $T_I\cup \EUF$-inconsistent}, 
i.e., the instantiation process terminates. This means that an interpolant can be extracted, using  a black-box quantifier-free interpolation algorithm for the weaker theory $T_I\cup \EUF$. In the simple case where $T_I$ is just the theory $TO$ of total orders, we shall prove in  Section~\ref{sec:to} that a \emph{quadratic} number of instantiations always suffices. In the general case, however, the situation is similar to the statement of Herbrand theorem: finitely many instantiations suffice 
to get an inconsistency proof in the weaker logical formalism, but a bound 
cannot be given. 

\subsection*{The proof.} Theorem~\ref{thm:al} is proved in a contrapositive way: we show that \emph{if a $T_I\cup \EUF$-inconsistency never arises, then $A\land B$ is \AXDTI-consistent}. This is proved in two steps: if $T_I\cup \EUF$-inconsistency does not arise, we produce two \AXDTI-models $\cA$ and $\cB$, where $\cA$ satisfies $A$ and $\cB$ satisfies $B$. Moreover, $\cA$ and $\cB$ are built up in such a way that they share the same \AXDTI-substructure. In the second step, we prove the amalgamation theorem for $\AXDTI$, so that the amalgamated model will produce the desired model of $A\land B$. In fact, the two steps are inverted in our exposition: we first prove the amalgamation theorem in Section~\ref{sec:semantic-arg} (Theorem~\ref{thm:axd_amalg}) and then our main theorem in Section~\ref{sec:interpolation} (Theorem~\ref{thm:al}).

\section{Embeddings}\label{sec:embeddings}

 We
preliminarily discuss the class of models of $\AXDTI$ and we make important clarifications about embeddings between such models.
A model $\cM$ of \AXEXTTI or of \AXDTI is \emph{functional} when the following conditions are satisfied: 
\begin{compactenum}
 \item[{\rm (i)}]
$\ARRAY^\cM$ is a subset of the set of all positive-support functions from $\INDEX^\cM$ to $\ELEM^\cM$
(a function $a$ is \emph{positive-support} iff $a(i)=\bot$ for every $i<0$); 
 \item[{\rm (ii)}] $rd$ is function application;
 \item[{\rm (iii)}] $wr$ is
the point-wise update operation (i.e., for $i\geq 0$, the function 
$wr(a,i,e)$ returns the same values as the function  $a$,
except at the index $i$ where it returns the
element $e$).  
\end{compactenum}
Because of the extensionality axiom, it can be shown
that every
model
is \emph{isomorphic  to a functional one}. 
For an array $a\in \INDEX^\cM$ in a functional model $\cM$ and for $i\in \INDEX^\cM$,  since $a$ is a function, we interchangeably use the notations $a(i)$ and $rd(a,i)$.
A functional model  $\cM$  is said to be \emph{full} iff $\ARRAY^\cM$ consists of \emph{all} the positive-support functions from $\INDEX^\cM$ to $\ELEM^\cM$. 

Let $a,b$ be
elements of $\ARRAY^\cM$ in a model $\cM$.  We say that
\emph{$a$ and $b$ are cardinality dependent} (in symbols, $\cM\models
\Vert a-b\Vert < \omega$) iff $\{i\in \INDEX^\cM \mid \cM\models rd(a,i)\neq
rd(b,i)\}$ is finite.  
Cardinality dependency in $\cM$ is obviously an equivalence relation, that we sometimes denote as $\sim_\cM$.

Passing to $\AXDTI$, a further  remark is in order: in a functional model $\cM$ of $\AXDTI$, the index $\diff(a,b)$ (if it exists) is uniquely determined: it must be the maximum index where $a,b$ differ (it is $0$ if $a=b$).
We say that $\diff(a,b)$ is \emph{defined} iff there is a maximum index where $a,b$ differ (or if $a=b$). An embedding $\mu: \cM\longrightarrow \cN$ between $\AXEXTTI$-models is said to be $\diff$-faithful iff whenever $\diff(a,b)$ is defined so is
$\diff(\mu(a),\mu(b))$ and it is equal to $\mu(\diff(a,b))$.
Since there might not be a maximum index where $a,b$ differ,  in principle it is not always possible to expand a functional model of $\AXEXTTI$ to a functional model of $\AXDTI$, 
keeping the set of indexes unchanged.  
Indeed, in order to do that in a $\diff$-faithful way, 
one needs to explicitly add to $\INDEX^{\cM}$ 
new indexes including at least  indexes representing
the missing maximum indexes 
where two given array differ.
This idea is used in the following lemma (proved 
in Appendix~\ref{sec:app}):

\begin{lemma}\label{lem:extension} For every index theory $T_I$,
 every  model of $\AXEXTTI$ has a $\diff$-faithful embedding into a model of $\AXD(T_I)$.
\end{lemma}


\section{Amalgamation}
\label{sec:semantic-arg}

We  now 
sketch the proof of 
the amalgamation property for \AXDTI. We recall that strong amalgamation holds for models of $T_I$  (see Definition~\ref{def:index}). 




\begin{theorem}\label{thm:axd_amalg}
 $\AXDTI$ enjoys the amalgamation property.
\end{theorem}

\begin{proof} 
Take two embeddings $\mu_1:\cN\longrightarrow \cM_1$ and
$\mu_2:\cN\longrightarrow \cM_2$.  As we know, we can
  suppose---w.l.o.g.---that $\cN, \cM_1, \cM_2$ are 
  functional models; in addition, via suitable renamings, we can freely suppose
  that $\mu_1, \mu_2$ restricts to inclusions for the sorts $\INDEX$
  and $\ELEM$, and that $(\ELEM^{\cM_1}\setminus \ELEM^\cN) \cap
  (\ELEM^{\cM_2}\setminus \ELEM^\cN)=\emptyset$,
  $(\INDEX^{\cM_1}\setminus \INDEX^\cN) \cap (\INDEX^{\cM_2}\setminus
  \INDEX^\cN)=\emptyset$.  To build the amalgamated model of $\AXDTI$, we first build a full model $\cM$ of $\AXEXTTI$ with $\diff$-faithful embeddings
  $\nu_1: \cM_1\longrightarrow \cM$ and $\nu_2: \cM_2\longrightarrow \cM$  
  such that $\nu_1\circ \mu_1=\nu_2\circ \mu_2$. If we succeed, the claim follows by Lemma~\ref{lem:extension}: 
  indeed, thanks to that lemma, we can embed in a $\diff$-faithful way $\cM$ (which is a model of $\AXEXTTI$) to a model $\cM'$ of $\AXDTI$, which is the required $\AXDTI$-amalgam.

  We take the $T_I$-reduct of $\cM$  to be a model supplied by the strong amalgamation property of $T_I$ (again, we can freely assume that the $T_I$-reducts of $\cM_1, \cM_2$ identically include in it);  we let $\ELEM^\cM$ to be
  $\ELEM^{\cM_1}\cup \ELEM^{\cM_2}$.  We need to define
  $\nu_i:\cM_i\longrightarrow \cM$ ($i=1,2$) in such a way that $\nu_i$ is $\diff$-faithful and
  $\nu_1\circ \mu_1=\nu_2 \circ \mu_2$. 
  We take the $\INDEX$ and the $\ELEM$-components of $\nu_1, \nu_2$ to be just identical inclusions.
  The only relevant point is
  the action of $\nu_i$ on $\ARRAY^{\cM_i}$: 
  since we have strong amalgamation for indexes,
  in
  order to define it, it is sufficient to extend any $a\in
  \ARRAY^{\cM_i}$ to 
  all 
  the indexes $k\in(\INDEX^{\cM}\setminus
  \INDEX^{\cM_i})$. For indexes $k\in(\INDEX^{\cM}\setminus
  (\INDEX^{\cM_1}\cup\INDEX^{\cM_2}))$ we can just put $\nu_i(a)(k)=\bot$.
  %
  If $k\in(\INDEX^{\cM}\setminus
  \INDEX^{\cM_i})$ and $k\in
  (\INDEX^{\cM_1}\cup\INDEX^{\cM_2})$, then $k\in(\INDEX^{\cM_{3-i}}\setminus
  \INDEX^{\cN})$;
  the definition for such $k$ is as follows:
  \begin{compactenum}
   \item[(*)] we let $\nu_i(a)(k)$ be equal to $\mu_{3-i}(c)(k)$, where $c$ is any array $c\in \ARRAY^\cN$ for which there is  $a'\in\ARRAY^{\cM_i}$ such that  $a\sim_{\cM_i} a'$ and such that the relation $k>\diff^{\cM_i}(a', \mu_i(c))$ holds in  $\INDEX^\cM$;\footnote{ 
   This should be properly written as $k>\nu_i(\diff^{\cM_i}(a', \mu_i(c)))$, however recall that the $\INDEX$-component of $\nu_i$ is identity, so the simplified notation is nevertheless correct. 
   } if such $c$ does not exist, then we put $\nu_i(a)(k)=\bot$.
  \end{compactenum}
Definition (*) is forced by some constraints that  $\nu_i(a)(k)$ must 
satisfy. 
Of course,  definition (*) itself needs to be justified: besides showing that it enjoys the required properties, we must also prove that it is well-given (i.e. that it does not depend on the selected $c$ and $a'$).  It is easy to see that, if the definition is correct, then we have $\nu_1\circ \mu_1=\nu_2 \circ \mu_2$; 
also, it is clear that
$\nu_i$ preserves read and write operations (hence, it is a homomorphism) and is injective. 
 For (i) justifying the definition of $\nu_i$ 
 and (ii) showing that it is also \diff-faithful,
  we need to show the following two claims 
  (the proof is not easy, see 
the 
 Appendix~\ref{sec:app} 
  for details)
 for arrays $a_1, a_2 \in \ARRAY^\cM_1$, for 
an index $k\in(\INDEX^{\cM_{2}}\setminus
  \INDEX^\cN)$ and for arrays $c_1, c_2 \in \ARRAY^\cN$ (checking the same facts in $\cM_2$ is symmetrical):
\begin{compactenum}
 \item[{\rm (i)}] if $a_1\sim_{\cM_1} a_2$ and $k>\diff^{\cM_1}(a_1, \mu_1(c_1))$, $k>\diff^{\cM_1}(a_2, \mu_1(c_2))$, then $\mu_{2}(c_1)(k)=\mu_{2}(c_2)(k)$.
 \item[{\rm (ii)}] if $k>\diff^{\cM_1}(a_1, a_2)$, then  $\nu_1(a_1)(k)=\nu_1(a_2)(k)$. 
 \eop
\end{compactenum}
\end{proof}

\section{Satisfiability}\label{sec:sat}

The key step of the interpolation algorithm that will be proposed in Section~\ref{sec:interpolation} depends upon the problem of checking satisfiability (modulo \AXDTI) of quantifier-free \formulae; this will be solved in the present section by adapting  instantiation techniques, like those from~\cite{BMS06}.

We define the \emph{complexity} $c(t)$ of a term $t$ as the number of function symbols occurring in $t$ (thus variables and constants have  complexity 0).
A \emph{flat} literal $L$ is a formula of the kind $x_1=t$ 
or $x_1\neq x_2$ or $R(x_1, \dots, x_n)$ or $\neg R(x_1, \dots, x_n)$, where the $x_i$ are variables, $R$ is a relation symbol, and $t$ is a term of complexity less or equal to 1. 
If $\cI$ is a  set of $T_I$-terms, an \emph{$\cI$-instance} of a universal formula of the kind  $\forall i\,\phi$ is a formula of the kind $\phi(t/i)$ for some $t\in \cI$.



A pair of sets of quantifier-free formulae $\Phi=(\Phi_1, \Phi_2)$ is a
 \emph{separated pair} 
iff
\begin{compactenum}
\item[{\rm (1)}] $\Phi_1$ contains  equalities of the form $\diff_k(a,b)=i$ and $a=wr(b,i,e)$; moreover if it contains the equality $\diff_k(a,b)=i$, it must also contain an equality of the form $\diff_l(a,b)=j$ for every $l<k$;
 \item[{\rm (2)}] $\Phi_2$ contains 
 Boolean combinations of $T_I$-atoms  and of atoms of the forms:
 \begin{equation}\label{phi2}
 rd(a,i)= rd(b,j), ~~
 rd(a,i)= e,~~
 e_1=e_2,
 \end{equation}
 where $a,b,i,j,e,e_1,e_2$ are variables or constants of the appropriate sorts.
\end{compactenum}
The separated pair is said to be finite iff $\Phi_1$ and $\Phi_2$ are both finite.

In practice, in a separated pair $\Phi=(\Phi_1, \Phi_2)$, reading $rd(a,i)$ as a functional application, it turns out that  \emph{ the \formulae from $\Phi_2$ can be translated into 
quantifier-free 
\formulae of the combined theory $T_I\cup\EUF$} (the array variables occurring in $\Phi_2$ are converted into free unary function symbols). $T_I\cup \EUF$  enjoys the decidability of the quantifier-free fragment and has quantifier-free interpolation because $T_I$ is an index theory (see Nelson-Oppen results~~\cite{NO79} and Theorem~\ref{thm:interpolation-amalgamation}): we  adopt a hierarchical approach (similar to~\cite{SS08,SS18}) and \emph{we rely  on satisfiability and interpolation algorithms for such a  theory as  black boxes}.

 Let $\cI$ be a set of $T_I$-terms and let $\Phi=(\Phi_1, \Phi_2)$ be a separated pair;
 we let $\Phi(\cI)=(\Phi_1(\cI), \Phi_2(\cI))$ be the smallest separated pair satisfying the following conditions:
 \begin{compactenum}
  \item[-] $\Phi_1(\cI)$ is equal to $\Phi_1$ and $\Phi_2(\cI)$ contains $\Phi_2$;
  \item[-] $\Phi_2(\cI)$ contains all $\cI$-instances of 
  the two \formulae 
  $$
  \forall i\; rd(\varepsilon,i)=\bot, ~ 
  \forall i\,(i<0 \to rd(a,i)=\bot),   
  $$
  where $a$ is any array variable occurring in $\Phi_1$ or $\Phi_2$;
  \item[-] if $\Phi_1$ contains the atom $a=wr(b,i,e)$ then $\Phi_2(\cI)$ contains \emph{all the $\cI$-instances of the 
  formulae~\eqref{eq:elwr}}; 
  \item[-] if $\Phi_1$ contains the conjunction $\bigwedge_{i=1}^l \diff_i(a,b)= k_i$, then 
  $\Phi_2(\cI)$ contains   the  formulae~\eqref{s0},~\eqref{s1}, \eqref{s2}, 
  \eqref{s3} as well as \emph{all $\cI$-instances of the formula~\eqref{s4}}.
 \end{compactenum}
For $M\in \mathbb N\cup \{\infty\}$, the \emph{$M$-instantiation} of $\Phi=(\Phi_1, \Phi_2)$ is the separated pair $\Phi(\cI_\Phi^M)=(\Phi_1(\cI_\Phi^M), \Phi_2(\cI_\Phi^M))$, where $\cI_\Phi^M$ is the set of $T_I$-terms of complexity at most $M$ built up from the index variables occurring in $\Phi_1, \Phi_2$.
The \emph{full instantiation} of $\Phi=(\Phi_1, \Phi_2)$  is the  separated pair 
$\Phi(\cI^{\infty}_\Phi)=(\Phi_1(\cI^{\infty}_\Phi), \Phi_2(\cI^{\infty}_\Phi))$ (which is usually not finite).
 A separated pair $\Phi=(\Phi_1, \Phi_2)$ is \emph{$M$-instantiated} iff $\Phi=\Phi(\cI_\Phi^M)$;
 it 
 is \AXDTI-satisfiable iff so it is the formula $\bigwedge \Phi_1\wedge \bigwedge \Phi_2$\footnote{
 This might be an infinitary formula if $\Phi$ is not finite.
 In such a case, satisfiability obviously  means that there is a model $\cM$ where we can assign values to all variables occurring in the \formulae from $\Phi_1\cup\Phi_2$
 in such a way that such \formulae become simultaneously true.
 }
\begin{example}\label{ex1}\emph{
  Let $\Phi_1$ contain the four atoms 
  $$
  \{~
  \diff(a, c_1) = i_1,~
  \diff(b, c_2)=i_1, ~
  a =wr(a_1, i_3, e_3),~
  a_1=wr(b, i_1, e_1) ~\}
  $$
 and let $\Phi_2$ be empty. Then $(\Phi_1, \Phi_2)$ is a separated pair; 
   0-instantiating it adds to $\Phi_2$    
   the following  
   formulae 
   (we delete those which are redundant)
  \begin{align*}
  & i_1\geq 0  & \\
  & rd(a,i_1)=rd(c_1,i_1)\to i_1=0   
  &   rd(b,i_1)=rd(c_2,i_1)\to i_1=0 \\  
   &  i_3>i_1 \to rd(a,i_3)=rd(c_1,i_3)  
  &  i_3>i_1 \to rd(b,i_3)=rd(c_2,i_3) \\  
   & i_3\geq 0 \to rd(a,i_3)=e_3   
  &  i_1\geq 0 \to rd(a_1,i_1)=e_1 \\   
   & i_1\neq i_3 \to rd(a,i_1)=rd(a_1,i_1)  
  & i_1\neq i_3 \to rd(a_1,i_3)=rd(b,i_3) 
  \end{align*}
  }
  \end{example}

The following results are proved   
in Appendix~\ref{sec:app}:

\begin{lemma}\label{lem:sat1}
 Let $\phi$ be a quantifier-free formula; then it is possible to compute finitely many finite 
 separation pairs $\Phi^1=(\Phi^1_1, \Phi^1_2), \dots, \Phi^n=(\Phi^n_1, \Phi^n_2)$
 such that $\phi$ is \AXDTI-satisfiable iff so is one of the $ \Phi^i$.
\end{lemma}

\begin{lemma}\label{lem:sat2}
 The following conditions are equivalent 
 for a finite separation pair $\Phi=(\Phi_1, \Phi_2)$:  
 \begin{compactenum}
 \item[{\rm (i)}] $\Phi$ is \AXDTI-satisfiable;
 \item[{\rm (ii)}] $\bigwedge \Phi_2(\cI_\Phi^0)$ is $T_I\cup \EUF$-satisfiable.
 \end{compactenum}
\end{lemma}

\begin{theorem}\label{thm:sat}
 The  $SMT(\AXDTI)$ problem is decidable for every
 index theory $T_I$ (i.e. for every theory satisfying Definition~\ref{def:index}). 
\end{theorem}  

Concerning the complexity of the above procedure, notice that the satisfiability of the quantifier-free fragment of common index theories (like \IDL, \LIA, \LRA) is decidable in NP; as a consequence, from the above proof we get  (for such index theories) also an NP bound for our $SMT(\AXDTI))$-problems
because 0-instantiation is clearly finite and  polynomial.
The fact that 0-instantiation suffices is a common feature of the above satisfiability procedure and of the satisfiability procedures from~\cite{BMS06}. Unfortunately, when coming to interpolation algorithms in the next section, there is no evidence that 0-instantiation suffices. 


\section{An interpolation algorithm}\label{sec:interpolation}

Since amalgamation is equivalent to quantifier-free interpolation for universal theories like \AXDTI (see Theorem~\ref{thm:interpolation-amalgamation}),  Theorem~\ref{thm:axd_amalg} ensures that \AXDTI has the quantifier-free interpolation property. However, the proof of Theorem~\ref{thm:axd_amalg} is not constructive, so in order to compute an interpolant  for  an 
\AXDTI-unsatisfiable conjunction like $\psi(\ux,\uy)\wedge \phi(\uy,\uz)$,  one should enumerate all quantifier-free \formulae $\theta(\uy)$ which    are logical  consequences of $\phi$ and are inconsistent with $\psi$ (modulo \AXDTI). Since the quantifier-free fragment of \AXDTI is decidable by Theorem~\ref{thm:sat}, this is an effective procedure and, since  interpolants of jointly unsatisfiable pairs of \formulae exist, it also terminates. However, such kind of an algorithm is not practical. 

In this section, we  improve the situation by supplying a better algorithm based on  instantiation (\`a-la-Herbrand). In the next section, 
using the results of the present section, for the special case where $T_I$ is just the theory of linear orders, we identify a complexity bound for this algorithm.

Our problem is the following: given two quantifier-free formulae $A$ and $B$ such that $A\wedge B$ is not satisfiable (modulo \AXDTI), to compute a quantifier-free formula $C$ such that $\AXDTI \models A\to C$, $\AXDTI \models C\wedge B\to \bot$ and such that $C$ contains only the variables (of sort \INDEX, \ARRAY, \ELEM) which occur both in $A$ and in $B$.

We call the variables occurring in both $A$ and $B$ \emph{common variables}, whereas the variables occurring in $A$ (resp. in $B$) are called \emph{$A$-variables} (resp. \emph{$B$-variables}). The same terminology applies to terms, atoms and formulae: e.g., a term $t$ is an $A$-term ($B$-term, common term) iff it is built up from $A$-variables 
($B$-variables, common variables, resp.).

The following operations can be freely performed (see~\cite{BGR14} or~\cite{axdiff} for details):
\begin{compactenum}
 \item[{\rm (i)}] pick an $A$-term $t$ and a fresh variable $a$ (of appropriate sort) and conjoin $A$ to $a=t$ ($a$ will be considered an $A$-variable from now on);
 \item[{\rm (ii)}] pick a $B$-term $t$ and a fresh variable $b$ (of appropriate sort) and conjoin $B$ to $b=t$ ($b$ will be considered a $B$-variable from now on);
 \item[{\rm (iii)}] pick a common term $t$ and a fresh variable $c$ (of appropriate sort) and conjoin both $A$ and $B$ to $c=t$ ($c$ will be considered a common variable from now on);
 \item[{\rm (iv)}] conjoin $A$ with some quantifier-free $A$-formula which is implied (modulo \AXDTI) by $A$;
 \item[{\rm (v)}] conjoin $B$ with some quantifier-free $B$-formula which is implied (modulo \AXDTI) by $B$.
\end{compactenum}
Operations (i)-(v) either add logical consequences or  explicit definitions that can be eliminated (if desired) after the final computation of the interpolant. In addition, notice that if $A$ is the form $A'\vee A''$ (resp. $B$ is of the form $B'\vee B''$) then from interpolants of $A'\wedge B$ and $A''\wedge B$ (resp.  of $A\wedge B'$ and $A\wedge B''$), we can recover an interpolant of $A\wedge B$ by taking disjunction (resp. conjunction).

Because of the above remarks, using the procedure in the proof of Lemma~\ref{lem:sat1},
both $A$ and $B$ are assumed to be given in the form of finite separated pairs.
 Thus $A$ is of the form $\bigwedge A_1\wedge \bigwedge A_2$, $B$ is of the form $\bigwedge B_1\wedge \bigwedge B_2$,
 for separated pairs $(A_1,A_2)$ and $(B_1, B_2)$.
 Also, by (iv)-(v) above, 
 $A$ and $B$ are assumed to be both 0-instantiated. 
 We call 
 $A$ (resp. $B$) the separated pair 
 $(A_1,A_2)$ (resp. $(B_1, B_2)$). We also use the letters $A_1, A_2, B_1, B_2$ both for sets of formulae and for the corresponding conjunctions; similarly, $A$ 
 represent both the pair $(A_1, A_2)$ and the conjunction $\bigwedge A_1\wedge \bigwedge A_2$ (and similarly for $B$). 

 The \formulae from $A_2$ and $B_2$ are \formulae from the signature of $T_I\cup \EUF$ (after rewriting terms of the kind $rd(a,i)$ to $f_a(i)$, where the $f_a$ are free function symbols).
 Of course, if $A_2 \wedge  B_2$ is $T_I\cup\EUF$-inconsistent, \emph{we can get our quantifier-free interpolant 
 by using our black box algorithm for interpolation in the weaker theory $T_I\cup\EUF$}: recall that $T_I\cup\EUF$ has quantifier-free interpolation because $T_I$ is an index theory and 
for Theorem~\ref{thm:interpolation-amalgamation}. The remarkable fact is that $A_2 \wedge  B_2$ always becomes  $T_I\cup\EUF$-inconsistent if \emph{sufficiently many \diff{s} among common array variables are introduced} and \emph{sufficiently many instantiations are performed}. 

 Formally, we shall  \emph{apply the loop below until 
 $A_2\wedge B_2$ becomes inconsistent}: the loop is justified by (i)-(v) above and 
 Theorem~\ref{thm:al}  guarantees that $A_2 \wedge  B_2$ eventually becomes inconsistent modulo $T_I\cup \EUF$, if $A\wedge B$ was originally inconsistent modulo \AXDTI. When $A_2\wedge B_2$  becomes inconsistent modulo $T_I\cup\EUF$, we can get our interpolant using the interpolation algorithm for $T_I\cup\EUF$.  
 [Of course, in the interpolant returned by $T_I\cup\EUF$, the extra variables introduced by the explicit definitions from (iii) above need to be eliminated.] 
 We need a counter $M$ recording how many times the Loop below has been executed (initially $M=0$). 
 \vskip 1mm\noindent
 \framebox{\textbf{Loop}} \emph{ (to be repeated until $A_2\wedge B_2$  becomes inconsistent modulo $T_I\cup\EUF$).
 Pick two distinct common \ARRAY-variables $c_1, c_2$ and $n\geq 1$ and s.t.
 no conjunct of the kind $\diff_n(c_1,c_2)=k$ occurs in both  $A_1$ and $B_1$ for some $n\geq 1$ (but s.t. for every $l<n$ there is a conjunct of the form $\diff_l(a,b)=k$ occurring in both $A_1$ and $B_1$). Pick also a fresh \INDEX constant $k_n$; conjoin 
 $\diff_n(c_1,c_2)=k_n$ to both $A_1$ and $B_1$; then $M$-instantiate both $A$ and $B$.
  Increase  $M$ to $M+1$.
 }
 
 Notice that the fresh index constants $k_n$ introduced during the loop are considered common constants (they come from explicit definitions like (iii) above) and so they 
 are considered in
 the $M$-instantiation of both $A$ and $B$.
 
 \begin{example}\label{ex2}\emph{  Let $A$ be the formula $\bigwedge \Phi_1$ from Example~\ref{ex1}
  and let $B$ be 
  $$
  i_1< i_2 ~\wedge ~i_2< i_3 ~\wedge~ rd(c_1,i_2)\neq rd(c_2, i_2)
  $$
  $B$ is 0-instantiated; 0-instantiating $A$ produces the  \formulae
  shown in Example~\ref{ex1}.
  The loop needs to be executed twice; it adds the literals
  $
  \diff_0(c_1,c_2)=k_0, \diff_1(c_1,c_2)=k_1
  $; 
  0-instantiation produces formulae $A_2$, $B_2$  whose conjunction is $T_I\cup
\EUF$-in\-con\-sistent (inconsistency can be tested via an SMT-solver like \textsc{z3} or \textsc{MathSat},  
see the ongoing implementation~\cite{Jose}).
 The related $T_I\cup
\EUF$-interpolant (once $k_0$ and $k_1$ are replaced by 
   $\diff_0(c_1,c_2)$ and $\diff_1(c_1,c_2)$, respectively) gives our \AXDTI-interpolant.  $\hfill\dashv$
  }
 \end{example}

 \begin{theorem}\label{thm:al}
  If  $A\wedge B$ is \AXDTI-inconsistent, 
  then the above loop terminates.
 \end{theorem}
\begin{proof} 
Suppose that the loop does not terminate and let $A'=(A_1', A_2')$ and 
$B'=(B_1', B_2')$ be the separated pairs obtained after infinitely many executions of the loop (they are the union of the pairs obtained in each step). Notice that both $A'$ and $B'$ are fully instantiated.\footnote{ 
On the other hand, the 
joined  pair $(A'_1\cup B'_1, A'_2\cup B'_2)$ is not even 0-instantiated.
} 
We claim that $(A', B')$ is \AXDTI-consistent (contradicting the assumption that $(A, B)$ was already \AXDTI-inconsistent). 

Since no contradiction was found, by compactness of first-order logic, $A'_2\cup B'_2$ has a $T_I\cup\EUF$-model $\cM$ (below we 
treat index and element variables occurring in $A, B$ as free constants and the array variables occurring in $A, B$ as free unary function symbols).
$\cM$ is a two-sorted structure (the sorts are \INDEX and \ELEM) endowed for every array  variable $a$ occurring in $A, B$ of a function $a^\cM:\INDEX^\cM \longrightarrow \ELEM^\cM$. In addition, $\INDEX^\cM$ is a model of $T_I$. 
%
 We  build three $\AXDTI$-structures $\cA, \cB, \cC$  
and two embeddings $\mu_1:\cC\longrightarrow \cA$, $\mu_2:\cC\longrightarrow\cB$ such that $\cA\models A'$, $\cB\models B'$ and such that  for every common variable $x$ we have $\mu_1(x^\cC)=x^\cA$ and $\mu_2(x^\cC)=x^\cB$. 
The consistency of $A'\cup B'$ then follows from the amalgamation Theorem~\ref{thm:axd_amalg}.
The  two structures $\cA, \cB$ are obtained by taking the full functional model induced by the restriction of $\cM$ to the interpretation  of $A$-terms and $B$-terms (respectively) of sort $\INDEX, \ELEM$ and then by applying Lemma~\ref{lem:extension}; the construction of $\cC$  requires some subtleties,  to be detailed  
in Appendix~\ref{sec:app}, 
where the full proof of the theorem is provided. \eop   
\end{proof}

\section{When indexes are just a total order}\label{sec:to}

Comparing the results from Sections~\ref{sec:interpolation} and~\ref{sec:sat}, a striking difference emerges: whereas variable and constant
instan\-tia\-tions are sufficient for satisfiability checking, our interpolation algorithm requires full instantiation over all common terms. Such a full instantiation might be quite impractical, especially in index theories like \LIA and \LRA (it is less  annoying in theories like \IDL: here all terms are of the kind $S^n(x)$ or $P^n(x)$, where $x$ is a variable or 0 and $S, P$ are the successor and the predecessor functions). The problem disappears in simpler theories like the theory of linear orders $TO$, where all terms are variables (or the constant 0).
%
%
Still, even in the case of $TO$, the proof of Theorem~\ref{thm:al} does not give a  bound for termination of the interpolation algorithm: we know that sooner or later an inconsistency will occur, but we do not know how many times we need to execute the main loop. 
We now improve the proof of Theorem~\ref{thm:al} by supplying the missing bound. In this section, the index theory is fixed to be $TO$ and we abbreviate $\AXD(TO)$ as $\AXD$. 
The full proof of the theorem below is in Appendix~\ref{sec:app}.

 \begin{theorem}\label{thm:allin}
  If  $A\wedge B$ is inconsistent modulo \AXD, then the above loop terminates
   in at most $(\frac{m^2-m} 2)\cdot (n+1)$ steps, where $n$ is the number of 
    the index variables occurring in $A,B$ and $m$ is the number of the common array variables.
 \end{theorem}

 \begin{proof}
 We sketch a proof of the theorem: the idea is that if
 %
after $N:=(\frac{m^2-m} 2)\cdot (n+1)$ steps
 no inconsistency occurs, then we can run the algorithm for infinitely many further steps without finding an inconsistency either.  
Let $A^N=(A_1^N, A_2^N)$ and 
$B^N=(B_1^N, B_2^N)$ be obtained after $N$-executions of the loop
and let  $\cM$ be a $TO\cup\EUF$-model of $A^N_2\wedge B^N_2$. 
Fix a pair of distinct common array variables $c_1,c_2$ to be handled in Step $N+1$;  since all pairs of common array variables have been examined in a fair way, $A_1^N$ and $B_1^N$ contain the atom  $\diff_{n+1}(c_1, c_2)=k_{n+1}$ (in fact $N:=(\frac{m^2-m} 2)\cdot (n+1)$ and $(\frac{m^2-m} 2)$ is the number of distinct unordered pairs of common array variables, so the pair $(c_1,c_2)$ has been examined more than $n$ times). In  $\cM$, some index variable $k_{l}$ for $l\leq k_{n+1}$, if not 
assigned 
to $0$, is assigned to an element $x$ which is different from the elements assigned to the $n$ variables occurring in $A, B$. This allows us to enlarge $\cM$
to a superstructure which is a model of $A^{N+1}_2\wedge B^{N+1}_2$ by 'duplicating' $x$.
Continuing in this way, we produce a chain of $TO\cup\EUF$-models  witnessing that we can run infinitely many steps of the algorithm without finding an inconsistency. 
\eop
 \end{proof}
 

\section{Conclusions and further work}\label{sec:conclusions}
We studied an extension of McCarthy theory of arrays with a maxdiff symbol. This symbol 
produces a much more expressive theory than  the theory of plain diff symbol already considered in the literature~\cite{axdiff,TW16}. 

We have also considered another strong enrichment, namely the combination with
arithmetic theories like 
$\IDL, \LIA, \LRA, \dots$ (all such theories are encompassed by the general notion of an `index theory').
Such a combination is non trivial because it is a non disjoint combination (the ordering relation  is in the shared signature) and does not fulfill the $T_0$-compatibility requirements of~\cite{Ghil05,GG18,GG17} needed in order to modularly import satisfiability and interpolation algorithms from the component theories.

The above enrichments 
come with a substantial cost: although decidability of satisfiability of quantifier-free formulae is not difficult to obtain, quantifier-free
 interpolation becomes challenging. In this paper, we proved that quantifier-free interpolants indeed do exist: the interpolation algorithm is indeed rather simple, but its justification comes via a complicated d\'etour involving semantic investigations on amalgamation properties.
 
The interpolation algorithm is 
based on hierarchic reduction to general quantifier-free interpolation in the index theory. The reduction requires the introduction of iterated diff terms and  a finite number of instantiations of the universal clauses associated to write and diff-atoms.
For the simple case where the index theory is just the theory of total orders, we were able to polynomially bound the depth of the iterated diff terms to be introduced as well as the number of instantiations needed. 
%
The main open problem we leave for future is the determination of analogous bounds for richer index theories. 

\bibliographystyle{plain}
\bibliography{mcmt,ghila,brutt,ranis}

\begin{thebibliography}{10}

\bibitem{Jose}
{AXDInterpolator}.
\newblock https://github.com/typesAreSpaces/AXDInterpolator.
\newblock Accessed: 2020-10-12.

\bibitem{ABGRS12a}
Francesco Alberti, Roberto Bruttomesso, Silvio Ghilardi, Silvio Ranise, and
  Natasha Sharygina.
\newblock Lazy abstraction with interpolants for arrays.
\newblock In {\em Proc. of LPAR-18}, volume 7180 of {\em LNCS}, pages 46--61.
  Springer, 2012.

\bibitem{ABGRS12}
Francesco Alberti, Roberto Bruttomesso, Silvio Ghilardi, Silvio Ranise, and
  Natasha Sharygina.
\newblock {SAFARI:} {SMT}-based abstraction for arrays with interpolants.
\newblock In {\em Proc. of {CAV}}, volume 7358 of {\em LNCS}, pages 679--685.
  Springer, 2012.

\bibitem{ABGRS14}
Francesco Alberti, Roberto Bruttomesso, Silvio Ghilardi, Silvio Ranise, and
  Natasha Sharygina.
\newblock An extension of lazy abstraction with interpolation for programs with
  arrays.
\newblock {\em Formal Methods Syst. Des.}, 45(1):63--109, 2014.

\bibitem{AGS14}
Francesco Alberti, Silvio Ghilardi, and Natasha Sharygina.
\newblock Booster: An acceleration-based verification framework for array
  programs.
\newblock In {\em Proc. of {ATVA}}, volume 8837 of {\em LNCS}, pages 18--23.
  Springer, 2014.

\bibitem{amalgam}
Paul~D. Bacsich.
\newblock Amalgamation properties and interpolation theorems for equational
  theories.
\newblock {\em Algebra Universalis}, 5:45--55, 1975.

\bibitem{BMS06}
Aaron~R. Bradley, Zohar Manna, and Henny~B. Sipma.
\newblock What's decidable about arrays?
\newblock In {\em Proc. of {VMCAI}}, volume 3855 of {\em LNCS}, pages 427--442.
  Springer, 2006.

\bibitem{brutto08}
Roberto Bruttomesso, Alessandro Cimatti, Anders Franz{\'{e}}n, Alberto Griggio,
  and Roberto Sebastiani.
\newblock The {MathSAT} 4 {SMT} solver.
\newblock In {\em Proc. of {CAV}}, volume 5123 of {\em LNCS}, pages 299--303.
  Springer, 2008.

\bibitem{axdiff}
Roberto Bruttomesso, Silvio Ghilardi, and Silvio Ranise.
\newblock Quantifier-free interpolation of a theory of arrays.
\newblock {\em Logical Methods in Computer Science}, 8(2), 2012.

\bibitem{BGR14}
Roberto Bruttomesso, Silvio Ghilardi, and Silvio Ranise.
\newblock Quantifier-free interpolation in combinations of equality
  interpolating theories.
\newblock {\em {ACM} Trans. Comput. Log.}, 15(1):5:1--5:34, 2014.

\bibitem{cade19}
Diego Calvanese, Silvio Ghilardi, Alessandro Gianola, Marco Montali, and Andrey
  Rivkin.
\newblock Model completeness, covers and superposition.
\newblock In {\em Proc.\ of {CADE}}, volume 11716 of {\em LNCS (LNAI)}, pages
  142--160. Springer, 2019.

\bibitem{IJCAR20}
Diego Calvanese, Silvio Ghilardi, Alessandro Gianola, Marco Montali, and Andrey
  Rivkin.
\newblock Combined covers and {Beth} definability.
\newblock In {\em Proc. of {IJCAR}}, volume 12166 of {\em LNCS (LNAI)}, pages
  181--200. Springer, 2020.

\bibitem{JAR21}
Diego Calvanese, Silvio Ghilardi, Alessandro Gianola, Marco Montali, and Andrey
  Rivkin.
\newblock Model completeness, uniform interpolants and superposition calculus
  (with applications to verificaton of data-aware processes).
\newblock {\em Journal of Automated Reasoning}, To appear.

\bibitem{CGU20}
Supratik Chakraborty, Ashutosh Gupta, and Divyesh Unadkat.
\newblock Verifying array manipulating programs with full-program induction.
\newblock In {\em Proc. of {TACAS}}, volume 12078 of {\em LNCS}, pages 22--39.
  Springer, 2020.

\bibitem{CK}
C.-C. Chang and H.~Jerome Keisler.
\newblock {\em Model Theory}.
\newblock North-Holland Publishing Co., Amsterdam-London, third edition, 1990.

\bibitem{cimatti-acm}
Alessandro Cimatti, Alberto Griggio, and Roberto Sebastiani.
\newblock Efficient generation of {Craig} interpolants in satisfiability modulo
  theories.
\newblock {\em {ACM} Trans. Comput. Log.}, 12(1):7:1--7:54, 2010.

\bibitem{cotton}
Scott Cotton and Oded Maler.
\newblock Fast and flexible difference constraint propagation for {DPLL(T)}.
\newblock In {\em Proc. of {SAT}}, volume 4121 of {\em LNCS}, pages 170--183.
  Springer, 2006.

\bibitem{Craig}
William Craig.
\newblock Three uses of the {H}erbrand-{G}entzen theorem in relating model
  theory and proof theory.
\newblock {\em J. Symbolic Logic}, 22:269--285, 1957.

\bibitem{FPMG19}
Grigory Fedyukovich, Sumanth Prabhu, Kumar Madhukar, and Aarti Gupta.
\newblock Quantified invariants via syntax-guided synthesis.
\newblock In {\em Proc. of {CAV}}, volume 11561 of {\em LNCS}, pages 259--277.
  Springer, 2019.

\bibitem{Ghil05}
Silvio Ghilardi.
\newblock Model theoretic methods in combined constraint satisfiability.
\newblock {\em J. Autom. Reasoning}, 33(3-4):221--249, 2004.

\bibitem{GG17}
Silvio Ghilardi and Alessandro Gianola.
\newblock Interpolation, amalgamation and combination (the non-disjoint
  signatures case).
\newblock In {\em Proc. of {FroCoS}}, volume 10483 of {\em LNCS (LNAI)}, pages
  316--332. Springer, 2017.

\bibitem{GG18}
Silvio Ghilardi and Alessandro Gianola.
\newblock Modularity results for interpolation, amalgamation and
  superamalgamation.
\newblock {\em Ann. Pure Appl. Logic}, 169(8):731--754, 2018.

\bibitem{GGK}
Silvio Ghilardi, Alessandro Gianola, and Deepak Kapur.
\newblock Compactly representing uniform interpolants for {EUF} using
  (conditional) {DAGS}.
\newblock Technical Report arXiv:2002.09784, arXiv.org, 2020.

\bibitem{CILC20}
Silvio Ghilardi, Alessandro Gianola, and Deepak Kapur.
\newblock Computing uniform interpolants for {EUF} via (conditional)
  {DAG}-based compact representations.
\newblock In {\em Proc. of {CILC}}, volume 2710 of {\em {CEUR} Workshop
  Proceedings}, pages 67--81. CEUR-WS.org, 2020.

\bibitem{GurfinkelSV18}
Arie Gurfinkel, Sharon Shoham, and Yakir Vizel.
\newblock Quantifiers on demand.
\newblock In {\em Proc. of {ATVA}}, volume 11138 of {\em LNCS}, pages 248--266.
  Springer, 2018.

\bibitem{HS18}
Jochen Hoenicke and Tanja Schindler.
\newblock Efficient interpolation for the theory of arrays.
\newblock In {\em Proc. of {IJCAR}}, volume 10900 of {\em LNCS (LNAI)}, pages
  549--565. Springer, 2018.

\bibitem{HuangInterpolant}
Guoxiang Huang.
\newblock Constructing {Craig} interpolation formulas.
\newblock In {\em Computing and Combinatorics {\em COCOON}}, volume 959 of {\em
  LNCS}, pages 181--190. Springer, 1995.

\bibitem{IIRS20}
Oren Ish{-}Shalom, Shachar Itzhaky, Noam Rinetzky, and Sharon Shoham.
\newblock Putting the squeeze on array programs: Loop verification via
  inductive rank reduction.
\newblock In {\em Proc. of {VMCAI}}, volume 11990 of {\em LNCS}, pages
  112--135. Springer, 2020.

\bibitem{kapurCC}
Deepak Kapur.
\newblock Shostak's congruence closure as completion.
\newblock In {\em Rewriting Techniques and Applications, 8th International
  Conference, RTA-97, Sitges, Spain, June 2-5, 1997, Proceedings}, pages
  23--37, 1997.

\bibitem{kapur}
Deepak Kapur.
\newblock Nonlinear polynomials, interpolants and invariant generation for
  system analysis.
\newblock In {\em Proc.\ of the 2nd International Workshop on Satisfiability
  Checking and Symbolic Computation co-located with {ISSAC}}, 2017.

\bibitem{kapurJSSC}
Deepak Kapur.
\newblock Conditional congruence closure over uninterpreted and interpreted
  symbols.
\newblock {\em J. Systems Science {\&} Complexity}, 32(1):317--355, 2019.

\bibitem{KMZ06}
Deepak Kapur, Rupak Majumdar, and Calogero~G. Zarba.
\newblock {Interpolation for Data Structures}.
\newblock In {\em Proc. of {SIGSOFT-FSE}}, pages 105--116. {ACM}, 2006.

\bibitem{KVGG19}
Hari Govind~Vediramana Krishnan, Yakir Vizel, Vijay Ganesh, and Arie Gurfinkel.
\newblock Interpolating strong induction.
\newblock In {\em Proc. of {CAV}}, volume 11562 of {\em LNCS}, pages 367--385.
  Springer, 2019.

\bibitem{mccarthy}
John McCarthy.
\newblock {Towards a Mathematical Science of Computation}.
\newblock In {\em IFIP Congress}, pages 21--28, 1962.

\bibitem{McM03}
Kenneth~L. McMillan.
\newblock Interpolation and {SAT}-based model checking.
\newblock In {\em Proc. of {CAV}}, volume 2725 of {\em LNCS}, pages 1--13.
  Springer, 2003.

\bibitem{McMillan05}
Kenneth~L. McMillan.
\newblock An interpolating theorem prover.
\newblock {\em Theor. Comput. Sci.}, 345(1):101--121, 2005.

\bibitem{McM06}
Kenneth~L. McMillan.
\newblock Lazy abstraction with interpolants.
\newblock In {\em Proc.\ of {CAV}}, volume 4144 of {\em LNCS}, pages 123--136.
  Springer, 2006.

\bibitem{mundici}
Daniele Mundici.
\newblock Craig's interpolation theorem, in computation theory.
\newblock {\em Atti della Accademia Nazionale dei Lincei. Classe di Scienze
  Fisiche, Matematiche e Naturali. Rendiconti, Serie 8}, 70(1):6--11, 1981.

\bibitem{NO79}
Greg Nelson and Derek~C. Oppen.
\newblock {Simplification by Cooperating Decision Procedures}.
\newblock {\em ACM Transactions on Programming Languages and Systems},
  1(2):245--57, 1979.

\bibitem{Pudlak97}
Pavel Pudl{\'{a}}k.
\newblock Lower bounds for resolution and cutting plane proofs and monotone
  computations.
\newblock {\em J. Symb. Log.}, 62(3):981--998, 1997.

\bibitem{SS08}
Viorica Sofronie{-}Stokkermans.
\newblock Interpolation in local theory extensions.
\newblock {\em Log. Methods Comput. Sci.}, 4(4), 2008.

\bibitem{SS18}
Viorica Sofronie{-}Stokkermans.
\newblock On interpolation and symbol elimination in theory extensions.
\newblock {\em Log. Methods Comput. Sci.}, 14(3), 2018.

\bibitem{TW16}
Nishant Totla and Thomas Wies.
\newblock Complete instantiation-based interpolation.
\newblock {\em J. Autom. Reasoning}, 57(1):37--65, 2016.

\bibitem{VG14}
Yakir Vizel and Arie Gurfinkel.
\newblock Interpolating property directed reachability.
\newblock In {\em Proc. of {CAV}}, volume 8559 of {\em LNCS}, pages 260--276.
  Springer, 2014.

\bibitem{yorsh}
Greta Yorsh and Madanlal Musuvathi.
\newblock A combination method for generating interpolants.
\newblock In {\em Proc. of {CADE}}, volume 3632 of {\em LNCS}, pages 353--368.
  Springer, 2005.

\end{thebibliography}

\appendix

\newpage

\section{Appendix}\label{sec:app}

\subsection{Preliminaries needed for technical proofs}

Given a signature $\Sigma$ and a $\Sigma$-structure $\cA$, we  
indicate with
$\Delta_{\Sigma}(\cA)$  
the \emph{diagram} of $\cA$:
this is the set of sentences obtained by first expanding $\Sigma$ with a fresh constant $\bar a$ for every
element $a$ from $\vert \cA\vert$ and then taking the set of ground $\Sigma\cup \vert \cA\vert$-literals which are true in $\cA$
(under the  natural expanded interpretation mapping $\bar a$ to $a$).\footnote{ 
As usual in model theory books, we won't distinguish anymore an element $a\in \vert \cA\vert$ from its name $\bar a$ in the expanded language 
$\Sigma\cup \vert \cA\vert$.
}
An easy but nevertheless important basic result (to be frequently used in our proofs), called 
\emph{Robinson Diagram Lemma}~\cite{CK},
says that, given any $\Sigma$-structure $\cB$, there is an embedding $\mu: \cA \longrightarrow \cB$
iff $\cB$ can be expanded to a  $\Sigma\cup \vert \cA\vert$-structure in such a way that it becomes a model of 
$\Delta_{\Sigma}(\cA)$.

\subsection{Embeddings}

We report here some important remarks on  models of \AXDTI and on embeddings that are missing in the main text. We then prove Lemma~\ref{lem:extension}.

Let $\cM$ be a model of \AXDTI or of \AXEXTTI.
 We say that some $a\in \ARRAY^\cM$ has \emph{finite support} if the set of $i\in \INDEX^\cM$ such that $\cM \models a(i)\neq \bot$ is finite. A functional model $\cM$ is said to be \emph{minimal} iff it consists of all the positive finite support functions from $\INDEX^\cM$ to $\ELEM^\cM$. A minimal functional model
  is  a model of both \AXEXTTI and of \AXDTI because its set of functions is closed under the $wr$ operation and $\diff(a,b)$ is already defined when $a,b$ have finite support (on the contrary, a full model might be a model only  of $\AXEXTTI$ and not also of $\AXDTI$ because there might not be a maximum index where two functions  $a$ and $b$ differ).

We show how any functional model $\cM$ of
\AXEXTTI (i.e. up to isomorphism, any model whatsoever) can be obtained
as a substructure of a full one. 
To this aim recall the definition of cardinality dependence from Section~\ref{sec:embeddings}.
%
%
In order to produce any such $\cM$, it is sufficient to take a
full model $\bar{\cM}$, to let $\INDEX^\cM \coincide \INDEX^{\bar{\cM}}$, $\ELEM^\cM \coincide
\ELEM^{\bar{\cM}}$, and to let $\ARRAY^\cM$ to be equal to any subset of
$\ARRAY^{\bar{\cM}}$ that is \emph{closed under cardinality dependence},
i.e.\ such that if $a\in \ARRAY^\cM$ and 
$a\sim_{\bar{\cM}} b$,
then $b$ is also in $\ARRAY^\cM$. If this happens, indeed $\ARRAY^{\bar{\cM}}$ is closed under the $wr$ operation and is a substructure. Viceversa, closure under the $wr$ operation implies closure under cardinality dependence: this is because if 
$a\sim_{\bar{\cM}} b$,
then $\cM\models b=wr(a, I, E)$,
where $I\coincide i_1, \ldots, i_n$ is a list of constants (naming elements of $\INDEX^\cM$), $E \coincide e_1, \ldots, e_n$ is a list of constants (naming elements of $\ELEM^\cM$) 
 and $wr(a, I, E)$ abbreviates the term $wr(wr(\cdots wr (a,
i_1, e_1) \cdots), i_n, e_n)$. Since $b$ is obtained from $a$ via iterated writings,  a subset closed under $wr$ operation and containing $a$ must also contain $b$.
In other words, functional
substructures $\cM$ of $\bar{\cM}$ with $\INDEX^\cM=\INDEX^{\bar{\cM}}$ and
$\ELEM^\cM= \ELEM^{\bar{\cM}}$ \emph{are in bijective correspondence with subsets of
$\ARRAY^{\bar{\cM}}$ closed under cardinality dependence}. 
The minimal model consists in selecting as subset just 
one  equivalence class (the 
equivalence class
$\varepsilon$ belongs to) of the  cardinality dependence relation. 

We prove here Lemma~\ref{lem:extension}, 
which is useful in order to build models of \AXDTI
out of models of \AXEXTTI:

\vskip 2mm\noindent
\textbf{Lemma~\ref{lem:extension}}  \emph{For every index theory $T_I$,
 every  model of $\AXEXTTI$ has a $\diff$-faithful embedding into a model of $\AXD(T_I)$.
}
\vskip 1mm

\begin{proof}
 Let $\cM$ be a model of $\AXEXTTI$ (we can freely suppose that it is functional).
 We show how to embed it in a $\diff$-faithful way in some $\cN$ so that $\diff(a,b)$ is defined for a given pair $a, b \in \ARRAY^\cM$.  The claim of the lemma follows by well ordering such pairs, repeating the construction for each pair by transfinite induction  and finally repeating the whole procedure $\omega$-times (notice that all this works because the axioms of $\AXEXTTI$ are universal and hence preserved by unions over chains).
 
 We suppose that $\ELEM^\cM$ contains at least two distinct elements $e_1, e_2$ and that $\diff(a,b)$ is not defined (otherwise there is nothing to do). We can freely take $e_2:=\bot$.
 
 If this is the case the set $I=\{i\in \INDEX^\cM \mid rd(a,i)\neq rd(b,i)\}$ does not have a maximum, hence in particular it is infinite. 
 We let $\downarrow\! I$ be the set of all $j$ such that there is $i\in I$ with $i\geq j$.
 By compactness of first-order logic (since $I$ is infinite), there is a model $\cA$ of $T_I$ extending the $T_I$-reduct of $\cM$ and containing an element $k_0$ 
 such that $i<k_0$ holds for $i\in \,\downarrow\! I$ and $k_0<i$ holds for $i\in \INDEX^\cM$ and $i\not\in \,\downarrow\! I$.
 %
 %
 Let $\ELEM^\cN$ be the same as 
 $\ELEM^\cM$; the $T_I$-reduct of $\cN$ (i.e., $\INDEX^\cN$) will be $\cA$;
 $\ARRAY^\cN$ will be the set of all functions from $\INDEX^\cN$ to 
 $\ELEM^\cN$ (thus $\cN$ is full).\footnote{However, the final model coming from our infinite iterations will not be full.} 
 
 We now define the embedding $\mu: \cM\longrightarrow \cN$. We let $\mu$ be the identical inclusion for $\INDEX$ and $\ELEM$; for $\ARRAY$ sort it is sufficient to specify the value $c(k)$ for all $c\in \ARRAY^\cM$ and all $k\in \INDEX^\cN\setminus \INDEX^\cM$ (then $\mu(c)$ is the same as $c$ extended to $\INDEX^\cN$ as specified). 
 The extension should be the same for $c_1, c_2$ such that $c_1\sim_\cM c_2$ (see the above observation on embeddings), it should be done in such a way that $\mu$ is $\diff$-faithful and it should be such that 
 $\diff(\mu(a),\mu(b))$ is defined in $\cN$ (actually $\diff(\mu(a),\mu(b))$ will be equal to $k_0$). In addition, we must have  $\varepsilon(i)=\bot=e_2$ for every $i\in \INDEX^\cN$. 
 We can freely assume  that at most only $b$ (and not also $a$) is such that there is $i\in I$, such that for all $j\in [i,k_0)$
 we have 
 $b(j)=\varepsilon(j)=\bot$ 
 (here $j\in [i,k_0)$ is the set of all $j\in \INDEX^\cM$ such that  $i\leq j<k_0$).
 
 We let $c(k)$ to be equal to $\bot$ for all $k\in \INDEX^\cN\setminus \INDEX^\cM$ different from $k_0$; for $k_0$, we let $c(k_0)$ to be equal
 to
 $e_1$ iff there is $i\in I$, such that for all $j\in [i,k_0)$
 we have $c(j)=a(j)$; we let $c(k_0)$ equal to $e_2=\bot$ otherwise. Since $I$ is infinite, it is easily checked that this definition satisfies the above requirements.
 \eop
\end{proof}

\subsection{Amalgamation of \AXDTI}

We report here the full proof of Theorem~\ref{thm:axd_amalg}.
In order to do so, we first need a
lemma  summarizing some `pseudo-metric' properties of $\diff$:

\begin{lemma}\label{lem:metric}
 The following sentences are logical consequences of $\AXD(T_I)$:
 \begin{eqnarray}
 \label{ax13}
  \forall x, y. & & \diff(x,y)\geq 0 \\
  \label{ax14}
  \forall x, y. & & \diff(x,y)=\diff(y,x) \\
  \label{ax15}
  \forall x, y, z. & & \max(\diff(x,y), \diff(y,z))\geq \diff(x,z)
  \end{eqnarray}
  where $\max$ denotes the maximum index of a pair (this is definable because $\leq$ is total).
\end{lemma}

\begin{proof}
 We only show the proof of the `triangular identity'~\eqref{ax15}. Suppose for instance that we have $\diff(x,y)\geq \diff(y,z)$; for $k>\diff(x,y)$ we have $rd(x,k)=rd(y,k)=rd(z,k)$. Let $k=\diff(x,z)$; if $x=z$ then the claim is trivial because $k=0$, otherwise we have $rd(x,k)\not\uguale rd(z,k)$. Thus, since  $k>\diff(x,y)$ implies $rd(x,k)= rd(z,k)$, we have $k\not >\diff(x,y)$, which means 
 $k\leq\diff(x,y)$, as required.
 \eop
\end{proof}

\vskip 2mm\noindent
\textbf{Theorem~\ref{thm:axd_amalg}}
 \emph{ $\AXDTI$ enjoys the amalgamation property.
}
\vskip 1mm

\begin{proof}
Take two embeddings $\mu_1:\cN\longrightarrow \cM_1$ and
$\mu_2:\cN\longrightarrow \cM_2$.  As we know, we can
  suppose---w.l.o.g.---that $\cN, \cM_1, \cM_2$ are 
  functional models; in addition, via suitable renamings, we can freely suppose
  that $\mu_1, \mu_2$ restricts to inclusions for the sorts $\INDEX$
  and $\ELEM$, and that $(\ELEM^{\cM_1}\setminus \ELEM^\cN) \cap
  (\ELEM^{\cM_2}\setminus \ELEM^\cN)=\emptyset$,
  $(\INDEX^{\cM_1}\setminus \INDEX^\cN) \cap (\INDEX^{\cM_2}\setminus
  \INDEX^\cN)=\emptyset$.  To build the amalgamated model of $\AXDTI$, we first build a full model $\cM$ of $\AXEXTTI$ with $\diff$-faithful embeddings
  $\nu_1: \cM_1\longrightarrow \cM$ and $\nu_2: \cM_2\longrightarrow \cM$  
  such that $\nu_1\circ \mu_1=\nu_2\circ \mu_2$. If we succeed, the claim follows by Lemma~\ref{lem:extension}: 
  indeed, thanks to that lemma, we can embed in a $\diff$-faithful way $\cM$ (which is a model of $\AXEXTTI$) to a model $\cM'$ of $\AXDTI$, which is the required $\AXDTI$-amalgam.

  We take the $T_I$-reduct of $\cM$  to be a model supplied by the strong amalgamation property of $T_I$ (again, we can freely assume that the $T_I$-reducts of $\cM_1, \cM_2$ identically include in it);  we let $\ELEM^\cM$ to be
  $\ELEM^{\cM_1}\cup \ELEM^{\cM_2}$.  We need to define
  $\nu_i:\cM_i\longrightarrow \cM$ ($i=1,2$) in such a way that $\nu_i$ is $\diff$-faithful and
  $\nu_1\circ \mu_1=\nu_2 \circ \mu_2$. 
  We take the $\INDEX$ and the $\ELEM$-components of $\nu_1, \nu_2$ to be just identical inclusions.
  The only relevant point is
  the action of $\nu_i$ on $\ARRAY^{\cM_i}$: 
  since we have strong amalgamation for indexes,
  in
  order to define it, it is sufficient to extend any $a\in
  \ARRAY^{\cM_i}$ to 
  all 
  the indexes $k\in(\INDEX^{\cM}\setminus
  \INDEX^{\cM_i})$.~\footnote{
  Strong amalgamation is required because it excludes that some index belongs to both $\INDEX^{\cM_i}$
  and to $\INDEX^{\cM_{3-i}}\setminus
  \INDEX^{\cN}$: for the latter indexes, we use definition (*) below and for the former indexes we just extend identically all $a\in \ARRAY^{\cM_i}$ (that is, for $k\in\INDEX^{\cM_i}$ and $a\in \ARRAY^{\cM_i}$, we put $\nu_i(a)(k)=a(k)$). The two definitions do not conflict because they apply to disjoint sets of indexes. 
  } For indexes $k\in(\INDEX^{\cM}\setminus
  (\INDEX^{\cM_1}\cup\INDEX^{\cM_2}))$ we can just put $\nu_i(a)(k)=\bot$.
  %
  If $k\in(\INDEX^{\cM}\setminus
  \INDEX^{\cM_i})$ and $k\in
  (\INDEX^{\cM_1}\cup\INDEX^{\cM_2})$, then $k\in(\INDEX^{\cM_{3-i}}\setminus
  \INDEX^{\cN})$;
  the definition for such $k$ is as follows:
  \begin{description}
   \item[(*)] we let $\nu_i(a)(k)$ be equal to $\mu_{3-i}(c)(k)$, where $c$ is any array $c\in \ARRAY^\cN$ for which there is  $a'\in\ARRAY^{\cM_i}$ such that  $a\sim_{\cM_i} a'$ and such that the relation $k>\diff^{\cM_i}(a', \mu_i(c))$ holds in  $\INDEX^\cM$;\footnote{ 
   This should be properly written as $k>\nu_i(\diff^{\cM_i}(a', \mu_i(c)))$, however recall that the $\INDEX$-component of $\nu_i$ is identity, so the simplified notation is nevertheless correct. 
   } if such $c$ does not exist, then we put $\nu_i(a)(k)=\bot$.
  \end{description}
Of course, the definition needs to be justified: besides showing that it enjoys the required properties, we must also prove that it is well-given (i.e. that it does not depend on the selected $c$ and $a'$).  It is easy to see that, if the definition is correct, then we have $\nu_1\circ \mu_1=\nu_2 \circ \mu_2$. 
Indeed, considering (*), if $a:=\mu_i(c)$ with $c \in \ARRAY^\cN$, we get for $k  \in(\INDEX^{\cM_{3-i}}\setminus
  \INDEX^{\cN})$ that $\nu_i(\mu_i(c))(k)=\mu_{3-i}(c)(k)$ holds by definition (since  $c$ itself can be taken as a representative), but we also have that $\nu_{3-i}(\mu_{3-i}(c))(k)=\mu_{3-i}(c)(k)$ since $\nu_{3-i}$ is just the
  identical extension when
  applied to indexes in $\INDEX^{\cM_{3-i}}$. 
  For  $k \in 
  \INDEX^{\cN}$, we have $\nu_1(\mu_1(c))(k)=\mu_1(c)(k)=c(k)=\mu_2(c)(k)=\nu_2(\mu_2(c))(k)$ because for these indexes $\nu_i$ is the identical extension and because the $rd$ operation (namely functional application) is preserved by $\mu_i$; finally, for  
   $k\in (\INDEX^{\cM}\setminus
  (\INDEX^{\cM_1}\cup\INDEX^{\cM_2}))$ we have $\nu_1(\mu_1(c))(k)=\bot=\nu_2(\mu_2(c))(k)$.
  This proves the required commutativity.   
Clearly, $\nu_i$ preserves read and write operations (hence, it is a homomorphism) and is also injective, being extended identically from indexes in $\INDEX^{\cM_i}$ to indexes in $\INDEX^{\cM}$.
 For (i) justifying the definition of $\nu_i$ and (ii) showing that it is also \diff-faithful, we need to prove the following two claims for arrays $a_1, a_2 \in \ARRAY^{\cM_1}$, for 
an index $k\in(\INDEX^{\cM_{2}}\setminus
  \INDEX^\cN)$ and for arrays $c_1, c_2 \in \ARRAY^\cN$ (checking the same facts in $\cM_2$ is symmetrical):
\begin{compactenum}
 \item[{\rm (i)}] if $a_1\sim_{\cM_1} a_2$ and $k>\diff^{\cM_1}(a_1, \mu_1(c_1))$, $k>\diff^{\cM_1}(a_2, \mu_1(c_2))$, then $\mu_{2}(c_1)(k)=\mu_{2}(c_2)(k)$.
 \item[{\rm (ii)}] if $k>\diff^{\cM_1}(a_1, a_2)$, then  $\nu_1(a_1)(k)=\nu_1(a_2)(k)$. 
\end{compactenum}
Point {\rm(i)} proves that the definition is well-given, since it does not depend on the choice of the representatives $a_1$ and $c_1$. 
Point {\rm(ii)} is necessary in order to guarantee that $\nu_1$ is $\diff$-faithful: indeed,  
(ii) guarantees that axiom~\ref{ax4} from Section~\ref{sec:tharr} applies not only to indexes $k\in\INDEX^{\cM_{1}}$ but also to  indexes $k\in\INDEX^{\cM_{2}}$ (for indexes  $k\in(\INDEX^{\cM}\setminus
  (\INDEX^{\cM_1}\cup\INDEX^{\cM_2}))$ we trivially have $\nu_1(a_1)(k)=\bot=\nu_1(a_2)(k)$).

\vspace{2mm}

\noindent\emph{Proof of}(i). The order is total, so suppose for instance that 
\begin{equation}\label{eq:geq}
\diff^{\cM_1}(a_1, \mu_1(c_1))\geq \diff^{\cM_1}(a_2, \mu_1(c_2)).
\end{equation}
Since  $a_1\sim_{\cM_1} a_2$ , we have that $a_1, a_2$ differ on at most finitely many indices and a subset of these indices comes from $\INDEX^\cN$. Let  $J=:\{j\in \INDEX^\cN\mid 
a_1(j)\neq a_2(j)~\&~j>\diff^{\cM_1}(a_1, \mu_1(c_1))\}$ and let $E:=\{a_1(j) \mid j\in J\}=\{c_1(j) \mid j\in J\}$ \footnote{Obviously, $\mu_1(c_1)(j)=c_i(j)$, since $j\in \INDEX^\cN$.}. 
Take the array $c\in \ARRAY^\cN$ defined as $wr(c_2, J, E)$, i.e. this is the array obtained by successively  overwriting $c_2$ in any $j\in J$ with $a_1(j)=c_1(j)$. Since $c\sim_\cN c_2$ and $k\in(\INDEX^{\cM_{2}}\setminus
  \INDEX^\cN)$, we have that 
\begin{equation}\label{eq:cc2}
 \mu_{2}(c)(k)=\mu_{2}(c_2)(k)
\end{equation}
(because $\mu_2$ is an embedding and as such preserves the writing operation).

We claim that 
\begin{equation}\label{eq:cl}
\diff^\cN(c_1,c)\leq \diff^{\cM_1}(\mu_1(c_1), a_1)
\end{equation}
i.e. that for every $j\in \INDEX^{\cN}$ such that $j>\diff^{\cM_1}(\mu_1(c_1), a_1)$, we have $c(j)= c_1(j)$.\footnote{
If~\eqref{eq:cl} does not hold, in fact there is  $j\in \INDEX^{\cN}$ such that $j>\diff^{\cM_1}(\mu_1(c_1), a_1)$ and $c(j)\neq c_1(j)$:  $\diff^\cN(c_1,c)$ is such a $j$.}
Pick such $j$ and, for the sake of contradiction, suppose that we have $c(j)\neq c_1(j)$; then, according to the definition of $c$, we must have 
\begin{equation}\label{eq:n1}
c_2(j)= c(j)~~
\end{equation}
(in fact, if $c_2(j)= c(j)$ does not hold, then $j$ must be one of the indexes where $c_2$ has been overwritten to get $c$ and in these indexes $c$ agrees with $c_1$, which is not the case since $c(j)\neq c_1(j)$). 
Since $j\in \INDEX^{\cN}$ and $j>\diff^{\cM_1}(\mu_1(c_1), a_1)$, the only possible reason why we have $c(j)\neq c_1(j)$, according to the definition of $c$, is because 
\begin{equation}\label{eq:n2}
a_1(j)= a_2(j)~~.
\end{equation}
By~\eqref{eq:geq}, we have $j>\diff^{\cM_1}(\mu_1(c_1), a_1)\geq
\diff^{\cM_1}(\mu_1(c_2),a_2)$, hence we get 
\begin{equation}\label{eq:n3}
c_1(j)= a_1(j)~~{\rm and}~~c_2(j)= a_2(j).
\end{equation}
Putting~\eqref{eq:n1},~\eqref{eq:n2},~\eqref{eq:n3} together, we get 
$c(j)= c_1(j)$, a contradiction. Thus the claim~\eqref{eq:cl} is established.

Now since $\diff^\cN(c_1,c)=\diff^{\cM_1}(\mu_1(c_1),\mu_1(c))\leq \diff^{\cM_1}(\mu_1(c_1), a_1)<k$, we get by transitivity $\diff^{\cM_2}(\mu_2(c_1),\mu_2(c))=\diff^\cN(c_1,c)<k$ which means 
\begin{equation}\label{eq:cc1}
 \mu_{2}(c)(k)=\mu_{2}(c_1)(k)
\end{equation}
Comparing~\eqref{eq:cc2} and~\eqref{eq:cc1}, we get $\mu_{2}(c_1)(k)=
\mu_{2}(c_2)(k)$, as required.

\vskip 2mm

\emph{Proof of}(ii). If $\diff(a_1,a_2)^{\cM_1}=0$ 
then $a_1\sim_{\cM_1} a_2$ and the claim is obvious by (*), so suppose $\diff^{\cM_1}(a_1,a_2)>0$. 
Suppose we apply (*) to find the value of $\nu_1(a_2)(k)$ for $k>\diff^{\cM_1}(a_1, a_2)>0$; we show that we can apply (*) to find the value of $\nu_1(a_1)(k)$ and that $\nu_1(a_1)(k)=\nu_1(a_2)(k)$.

According to (*), 
there are 
$b_2\in \ARRAY^{\cM_1}$ 
and $c\in \ARRAY^\cN$  
such that
$a_2\sim_{\cM_1} b_2$ and  
 $k>\diff^{\cM_1}(b_2, \mu_1(c))$ - with $\nu_1(a_2)(k)$ defined to be $\mu_2(c)(k)$.
  Since $a_2\sim_{\cM_1} b_2$, 
 the arrays $a_2$ and $b_2$ differ on finitely many indices from $\INDEX^{\cM_1}$ and let $I=i_1, \dots, i_n$ be the subset of such indices which are bigger than $\diff^{\cM_1}(a_1, a_2)$. 
  Let also  $b_1$ be obtained from $a_1$ over-writing $a_1$ on indices $i_1, \dots, i_n$ with $b_2(i_1), \dots b_2(i_n)$ respectively; using a self-explaining notation, we have $b_1:=wr(a_1, I, rd(b_2,I))$. 
 We clearly have $a_1\sim_{\cM_1} b_1$ and we claim that 
 \begin{equation}\label{eq:diffm}
  \diff^{\cM_1}(b_1, b_2)\leq\diff^{\cM_1}(a_1, a_2)~~.
 \end{equation}
To prove~\eqref{eq:diffm}, we take $l>\diff(a_1,a_2)$ and we show that we have $b_1(l)=b_2(l)$. Suppose for contradiction that $b_1(l)\neq b_2(l)$ for such $l>\diff(a_1,a_2)$; since $b_1=wr(a_1, I, rd(b_2,I))$, we must have $a_1(l)=b_1(l)$ and $l\not\in I$, that is (keeping in mind that $l>\diff(a_1,a_2)$ and the definition of $I$) $l$ is an index such that $a_2(l)= b_2(l)$ and $a_1(l)=a_2(l)$; putting all these relations together, we get $b_1(l)=b_2(l)$, contradiction.

 
 Having found $b_1\sim_{\cM_1}a_1$ such that $\diff^{\cM_1}(b_1, b_2)\leq \diff^{\cM_1}(a_1, a_2)$, we proceed as follows.
 Recall that, by hypothesis of {\rm(ii)}, $k>\diff^{\cM_1}(a_1, a_2)$,  
  and that $k>\diff^{\cM_1}(b_2, \mu_1(c))$, so
 we have
 \begin{eqnarray*}
 &
 k>\max(\diff^{\cM_1}(a_1, a_2),\diff^{\cM_1}(b_2, \mu_1(c)))\geq~~~~~~~~~~~~~~~~~~~~~~~~~
 \\ &
 \geq\max(\diff^{\cM_1}(b_1, b_2),\diff^{\cM_1}(b_2, \mu_1(c)))\geq 
\diff^{\cM_1}(b_1, \mu_1(c))
\end{eqnarray*}
where in the last inequality 
we applied Lemma~\ref{lem:metric}. Thus, according to the definition of $\nu_1$ via (*) (with $b_1$ as $a'$), we have 
$\nu_1(a_1)(k)=\mu_{2}(c)(k)=\nu_1(a_2)(k)$, as required.
\eop
\end{proof}

\subsection{Satisfiability}

We report here the proof of the two technical lemmas that are used in the proof of Theorem~\ref{thm:sat};
we recall that Theorem~\ref{thm:sat} follows from these lemmas applying Nelson Oppen combination result~\cite{NO79} to $T_I\cup \EUF$: according to such result, the SMT satisfiability problem is decidable for a union of stably infinite, signature disjoint theories whose SMT satisfiability problems are separately decidable.

\vskip 2mm\noindent
\textbf{Lemma~\ref{lem:sat1}}
 \emph{
 Let $\phi$ be a quantifier-free formula; then it is possible to compute finitely many finite 
 separated pairs $\Phi^1=(\Phi^1_1, \Phi^1_2), \dots, \Phi^n=(\Phi^n_1, \Phi^n_2)$
 such that $\phi$ is \AXDTI-satisfiable iff so is one of the $ \Phi^i$.
}

 \begin{proof} We first flatten all atoms from $\phi$ by repeatedly abstracting out subterms 
 (to abstract out a subterm $t$, we introduce a fresh variable $x$ and update $\phi$ to
 $x=t\wedge  \phi(x/t)$); then we remove all atoms of the kind $a=b$ occurring in $\phi$ by replacing them by the equivalent formula~\eqref{eq:eleq}, namely
 $$
  \diff(a,b)=0 \wedge rd(a,0)=rd(b,0)~.
 $$
  Finally, we put $\phi$ in disjunctive normal form and extract equisatisfiable separated pairs from each disjunct 
  (to this aim we might need to introduce some atoms like $\diff_l(a,b)=k_l$, with fresh $k_l$, as conjuncts). 
  \eop
 \end{proof}

\vskip 2mm\noindent
\textbf{Lemma~\ref{lem:sat2}}
 \emph{
 TFAE for a finite separated pair $\Phi=(\Phi_1, \Phi_2)$:  
 \begin{compactenum}
 \item[{\rm (i)}] $\Phi$ is \AXDTI-satisfiable;
 \item[{\rm (ii)}] $\bigwedge \Phi_2(\cI_\Phi^0)$ is $T_I\cup \EUF$-satisfiable.
 \end{compactenum}
} 
 \begin{proof}
  Obviously (i) $\Rightarrow$ (ii). 
  Assume that we have a $T_I\cup \EUF$-model satisfying $\bigwedge \Phi_2(\cI_\Phi^0)$. This means that there is a $T_I\cup \EUF$-model $\cM$ where we can assign values to all variables occurring in the \formulae from $\Phi_2(\cI_\Phi^0)$
 in such a way that such \formulae become simultaneously true (of course, to make all this meaningful, we 
 replace  array variables $a$ by free unary function symbols $f_a$ and rewrite the terms $rd(a,i)$
 to $f_a(i)$). We can freely assume that in $\cM$ we have $f_a(k)=\bot$ for every $k\in \INDEX^\cM$ different from the values assigned to the index variables  $\cI_\Phi$: changing the values of such $f_a(k)$ would not affect satisfiability
 of the formulae in $\Phi_2(\cI_\Phi^0)$ by the shape of the atoms occurring in these \formulae. In this way, the $f_a$ are positive finite support
functions. The minimal functional model over $\INDEX^\cM$ and $\ELEM^\cM$ will satisfy 
also $\bigwedge\Phi_1$ by construction. Thus $\bigwedge \Phi_1\wedge \bigwedge \Phi_2$ is \AXDTI-satisfiable.
\eop
 \end{proof}

 \subsection{Algorithm}
 
 We report here in full detail the proof of Theorem~\ref{thm:al}.
 
\vskip 2mm\noindent
\textbf{Theorem~\ref{thm:al}}
  \emph{ If  $A\wedge B$ is inconsistent modulo \AXD, then the
  Loop from Section~\ref{sec:interpolation}
  terminates.\footnote{It goes without saying that the Loop must be executed in a fair way, i.e. that  for every triple $c_1,c_2, k$ there should be a step where the triple is taken into consideration (if the algorithm does not stop earlier).}
}
\vskip 1mm

\begin{proof} 

Suppose that the loop does not terminate and let $A'=(A_1', A_2')$ and 
$B'=(B_1', B_2')$ be the separated pairs obtained after infinitely many executions of the loop (they are the union of the pairs obtained in each step). Notice that both $A'$ and $B'$ are fully instantiated.\footnote{ 
On the other hand, the 
joined separated pair $(A'_1\cup B'_1, A'_2\cup B'_2)$ is not even 0-instantiated.
} We claim that $(A', B')$ is \AXDTI-consistent (contradicting the assumption that $(A, B)$ was 
\AXDTI-inconsistent).

Since no contradiction was found, by compactness of first-order logic, $A'_2\cup B'_2$ has a $T_I\cup\EUF$-model $\cM$ (below we 
treat index and element variables occurring in $A, B$ as free constants and the array variables occurring in $A, B$ as free unary function symbols).
$\cM$ is a two-sorted structure (the sorts are \INDEX and \ELEM) endowed for every array  variable $a$ occurring in $A, B$ of a function $a^\cM:\INDEX^\cM \longrightarrow \ELEM^\cM$. In addition, $\INDEX^\cM$ is a model of $T_I$. 
 We shall build three $\AXDTI$-structures $\cA, \cB, \cC$  
and two embeddings $\mu_1:\cC\longrightarrow \cA$, $\mu_2:\cC\longrightarrow\cB$ such that $\cA\models A'$, $\cB\models B'$ and such that  for every common variable $x$ we have $\mu_1(x^\cC)=x^\cA$ and $\mu_2(x^\cC)=x^\cB$. 
The consistency of $A'\cup B'$ then follows from the amalgamation Theorem~\ref{thm:axd_amalg}.
In view of Lemma~\ref{lem:extension},  
only  $\cC$ must be an \AXDTI-model: 
$\cA$ and $\cB$ need only to be $\AXEXTTI$-models, in case $\mu_1, \mu_2$ are \diff-faithful, 
$\diff(a_1, a_2)$ is defined in $\cA$ for every pair $a_1,a_2$ of $A$-variables of sort \ARRAY, and
$\diff(b_1, b_2)$ is defined in $\cB$ for every pair $b_1,b_2$ of $B$-variables of sort \ARRAY.   


We take as $\cA$ the full functional \AXEXTTI-structure having as $\INDEX^\cA$ and $\ELEM^\cA$ the restrictions of $\INDEX^\cM$ and of $\ELEM^\cM$ to the elements of the kind $t^\cM$, where $t$ is  an $A'$-term. Functions and relation symbols from the signature of $T_I$ are interpreted as restrictions of their interpretations in $\cM$. Since $T_I$ is an index theory, it is universal and hence closed under taking substructures, so the resulting model is a model of $\AXEXTTI$. 
We assign to an $A$-variable $a$ of sort \ARRAY the function that maps, for every $A'$-term 
in the signature of $T_I$, the element $t^\cA:=t^\cM\in \INDEX^\cA$ to
$a^\cM(t^\cM)$. Since $A'$ is $n$-instantiated for every $n$ and since $\cM\models A'_2$, by Lemma~\ref{lem:univinst} 
we have that $\diff(a_1, a_2)$ is defined in $\cA$ for every pair $a_1,a_2$ of $A$-variables of sort \ARRAY and that $\cA\models A'$.
The \AXEXTTI-structure $\cB$ and the assignment to the $B$-variables are  defined analogously. 

The definition of $\cC$ is more subtle. Again we take as $\INDEX^\cC$ and $\ELEM^\cC$ the restrictions of $\INDEX^\cM$ and of $\ELEM^\cM$ to the elements of the kind $t^\cM$, where $t$ is  a common term (i.e. it is both an $A'$- and a $B'$-term); again function and relation symbols in the signature of $T_I$ are interpreted by restriction. We take as $\ARRAY^\cC$ the set of functions $d:\INDEX^\cC\longrightarrow \ELEM^\cC$ such that there is a common array variable $c$ such that $d$ differ only by finitely many indices from  the restriction of  $c^\cM$ to
$\INDEX^\cC$ in its domain and to $\ELEM^\cC$ in its codomain. Obviously, for a common array variable $c$, we let $c^\cC$ be the restriction of $c^\cM$ to
$\INDEX^\cC$ in the domain and to $\ELEM^\cC$ in the codomain.

To simplify notation, from now on, for a common index variable $i$, we write $i^\cM=i^\cA=i^\cB=i^\cC$ just as $i$.

We first show that $\diff$ is totally defined in $\cC$ (so that $\cC$ is a model of \AXDTI).
Notice first that $\diff^\cC_n(c_1^\cC,c_2^\cC)$ is defined for all common array variables $c_1,c_2$ and for all $n\geq 1$: this is shown as follows.
The full instantiation of the clauses produced by the formulae $\diff_n(c_1,c_2)=k_n$ in the loop implies that in $\cM$ (hence also in the substructure $\cC$) we must have 
$\cM \models c_1(t)=c_2(t)\vee \bigvee_{i=1}^n t=k_i$ for all common index terms $t$ such that $\cM\models t>k_n$ (in fact, the $k_i$'s are common variables too, according to 
our interpolation algorithm as specified in Section~\ref{sec:interpolation})
and also that $\cM \models c_1(k_n)=c_2(k_n)\to k_n=0$. This means that $\cC\models \diff_n(c_1,c_2)=k_n$, because  $\INDEX^\cC$ is formed precisely by the elements of the kind $t^\cM$ for common index terms $t$.

Take now $d_1, d_2\in \ARRAY^\cC$; we have $d_1\sim_{\cC} c_1^\cC$ and $d_2\sim_{\cC} c_2^\cC$ for some common array variables $c_1, c_2$. 
This means that $d_1$ differs from $c_1^\cC$ for the indexes in a finite set  $I_1$ and $d_2$ differs from $c_2^\cC$ for the indexes in a finite set  $I_2$.
Let $n$ big enough so that we have $\cC\models \diff_n(c_1,c_2)=k_n$ and either $k_n$ is $0$ or
$k_n\not \in I_1\cup I_2$.  If $k_n$ is 0, then $d_1\sim_{\cC}c_1^\cC \sim_{\cC}c_2^\cC
\sim_{\cC} d_2$, hence $\diff^\cC(d_1,d_2)$
is clearly defined. If $k_n$ does not belong to $I_1\cup I_2$ and is bigger than 0,   notice that
$d_1(k_n)=c^\cC_1(k_n)\neq c^\cC_2(k_n)=d_2(k_n)$; moreover
above $k_n$ there are only finitely many indexes where $d_1$ and $d_2$ can differ
(because if $i>k_n$ is such that $d_1(i)\neq d_2(i)$, then we must have 
$i\in I_1\cup I_2\cup \{k_1, \dots, k_{n-1}\}$). This means that 
$\diff^\cC(d_1,d_2)$
is  defined in this case too and it belongs to $I_1\cup I_2\cup \{k_1, \dots, k_{n-1}, k_n\}$.

It remains only to define the embeddings $\mu_1$ and $\mu_2$. We show the definition of $\mu_1$ (the definition of $\mu_2$ is analogous). 
Obviously, $\mu_1$ acts as an inclusion for $\INDEX$ and $\ELEM$ sorts. To define the \ARRAY-component of $\mu_1$, we make
a preliminary observation concerning two common array variables $c_1, c_2$.
The observation is  that if $c_1^\cC\sim_\cC c_2^\cC$, then
for every $i\in \INDEX^\cA\setminus\INDEX^\cC$, we have that $c_1^\cM(i)=c_2^\cM(i)$ and consequently also $c_1^\cA(i)=c_2^\cA(i)$.
This is because if $c_1^\cC\sim_\cC c_2^\cC$, then for  $n$ big enough we must have $\cM\models \diff_n(c_1,c_2)=k_n$ and $\cM\models c_1(k_n)=c_2(k_n)$ (the $k_n$ are common variables, so $c_1^\cC(k_n)$ and $c_2^\cC(k_n)$ cannot differ for infinitely many $n$ if the $k_n$  are all distinct, given that $c_1^\cC\sim_\cC c_2^\cC$); by Lemma~\ref{lem:elim}\eqref{s2}-\eqref{s3} and 0-instantiation,
 we must get $\cM \models k_n=0$ and so $\cM \models c_1(i)=c_2(i)$ for the  $i\in \INDEX^\cA\setminus\INDEX^\cC$ by Lemma~\ref{lem:elim}\eqref{s4} and full instantiation ($i$ cannot coincide with any of the $k_n$ because it is not common).
%
Thus we can define $\mu_1(d)$ for every $d\in \ARRAY^\cC$ as follows: pick \emph{any} common variable $c$ such that $d\sim c^\cC$  and let $\mu_1(d)(i)=c^\cM(i)$ for all $i\in \INDEX^\cA\setminus\INDEX^\cC$ (for $i\in \INDEX^\cC$, we obviously put $\mu_1(d)(i)=d(i)$). 

In this way it is clear that $\mu_1$ preserves $rd$ and $wr$ operations. 
Since 
we have that $\mu_1(c^\cC)=c^\cA$ for all common array variables $c$, we only have to prove that $\diff$ is preserved. We show that, for $d_1,d_2\in \ARRAY^\cC$ and for $i\in \INDEX^\cA\setminus\INDEX^\cC$ such that $\mu_1(d_1)(i)\neq \mu_1(d_2)(i)$, there always is some $k\in \INDEX^\cC$ such that $d_1(k)\neq d_2(k)$ and $k>i$ (thus $\diff^\cA(\mu_1(d_1), \mu_1(d_2))$ must be  $\diff^\cC(d_1, d_2)$). Now, if $\mu_1(d_1)(i)\neq \mu_1(d_2)(i)$ for some  $i\in \INDEX^\cA\setminus\INDEX^\cC$, this can happen only if $d_1\not\sim_\cC d_2$ according to the definition of $\mu_1$. Thus there are common array variables $c_1, c_2$ such that $c_1^\cC\sim_\cC d_1$,
$ c_2^\cC \sim_\cC d_2$ and $c_1^\cA\not \sim_\cC c_2^\cA$ with $c_1^\cA(i)=\mu_1(d_1)(i)$ and
$c_2^\cA(i)=\mu_1(d_2)(i)$, thus $c_1^\cA(i)\neq c_2^\cA(i)$. Since $c_1^\cA\not \sim_\cC c_2^\cA$, the loop produces infinitely many $k_n$  such that $c_1^\cC(k_n)\neq c_2^\cC(k_n)$ via $\diff_k(c_1,c_2)=k_n$; we must have that $k_n>i$ holds in $\cA$ for all such $k_n$ because of full instantiation and because $i$ belongs to $\INDEX^\cA\setminus\INDEX^\cC$, so it cannot be equal to any of the $k_n$. For infinitely many of those $k_n$ we have $c_1^\cC(k_n)=d_1(k_n)$ and  $c_2^\cC(k_n)=d_2(k_n)$, so there is certainly an $n$ such that $k_n>i$ and $d_1(k_n)\neq d_2(k_n)$, as required.
\eop
\end{proof}

\subsection{When indexes are just a total order}

We report here the full proof of Theorem~\ref{thm:allin}.

\vskip 2mm\noindent
\textbf{Theorem~\ref{thm:allin}}
  \emph{ If  $A\wedge B$ is inconsistent modulo \AXD, then 
  the Loop from Section~\ref{sec:interpolation}
  terminates
   in at most $(\frac{m^2-m} 2)\cdot (n+1)$ steps, where $n$ is the number of 
    the index variables occurring in $A,B$ and $m$ is the number of the common array variables.
}
\vskip 1mm
 
 \begin{proof}
 To prove the theorem, it is sufficient to show that if 
after $N:=(\frac{m^2-m} 2)\cdot (n+1)$ steps no inconsistency occurs, then we can run the algorithm for infinitely many further steps without finding an inconsistency either.


Let $A^N=(A_1^N, A_2^N)$ and 
$B^N=(B_1^N, B_2^N)$ be the pairs obtained after $N$-executions of the loop
and let  $\cM$ be a $TO\cup\EUF$-model of $A^N_2\wedge B^N_2$. 
Since 0-instantiation is the same as $N$-instantiation
(and as full instantiation) in \AXD, by restricting to a suitable substructure,   we can freely suppose that 
$\INDEX^\cM$ contains just 0 and the elements assigned to  
the index constants occurring in $A^N$ and $ B^N$. Thus $\INDEX^\cM$ is a finite set.

Notice that  after $N$ steps, if all pairs of common array variables have been examined in a fair way, for every pair of distinct common array variables $c_1, c_2$, we have that  $\diff_{n+1}(c_1, c_2)=k_{n+1}$ occurs as a conjunct in both $A_1^N$ and $B_1^N$ (we recall that $n$ is the number of index variables occurring in $A, B$).
Fix such a pair of distinct common array variables $c_1,c_2$ to be handled in Step $N+1$.  Now, by the definition of 
instantiation and by~\eqref{s0}, \eqref{s2}, either 
$k_{n+1}=0$ or there must be an index $k_l$ (with $l\leq n+1$) such that the element  $x=k_l^\cM$ assigned to it in $\cM$ is different from all the 
elements 
assigned to
the index constants occurring in $A, B$~\footnote{ 
In fact  $\frac{m^2-m} 2$ is the number of distinct unordered pairs of common array variables;
after $N=(\frac{m^2-m} 2) \cdot (n+1)$ fair iterations, for every such pair $c_1, c_2$ of common array variables, we added $k_l=\diff_l(c_1,c_2)$ more than $n$ times.
}
The former case is trivial (nothing happens executing 
further steps
of the algorithm relatively to $c_1,c_2$). 

In the latter case, we   enlarge $\cM$ to a superstructure $\cM^+$ by `duplicating'  $x$; what we do is to   add to $\INDEX^{\cM}$ a fresh element $y$ such that in $\cM^+$ the following happen: 
(i) $x<y$; 
(ii) $j< y$ holds iff $ j< x$ holds in $\cM$ for $j\in \INDEX^\cM, j\neq x$; 
(iii) $y<j$ holds iff $x<j$  holds in $\cM$ for $j\in \INDEX^\cM, j\neq x$;  
(iv) for all array variables $d$ occurring in $A,B$, we have $d^{\cM^+}(j)=d^\cM(j)$
for all $j\in \INDEX^\cM, j\neq y$ and $d^{\cM^+}(y)=d^\cM(x)$.

In $\cM^+$, we assign to every index or element variable occurring in $A,B$ the same value it had in $\cM$. The further common variables $h_s$ introduced by the algorithm in Syeps $1, \dots, N+1$ and appearing in \formulae like $\diff_1(c'_1,c'_2)=h_1\, \wedge \cdots \wedge\, \diff_r(c'_1,c'_2)=h_r$ are assigned to 
appropriate elements in $\INDEX^{\cM^+}$ so that $h_1^{\cM^+}, \dots, h_r^{\cM^+}$ are the $r$-th largest elements where 
$(c'_1)^{\cM^+}$ and $(c'_2)^{\cM^+}$ differ.~\footnote{ 
What may  happen when passing from $\cM$ to $\cM^+$ is that some element playing the role in $\cM$ of $\diff_s(c'_1, c'_2)$  plays now the role  of $\diff_{s+1}(c'_1, c'_2)$ in $\cM^+$ because $y$ has been inserted above it.
} Notice that in this way we can assign values to the variables $k_1, \dots, k_{n+1}, k_{n+2}$ relative to the pair $c_1, c_2$, including the new variable $k_{n+2}$
representing the $n+2$-th iterated $\diff$ of $c_1, c_2$.

It remains to check that $\cM^+$ is a model of $A_2^{N+1}\wedge B_2^{N+1}$.
The only problem concerns the $0$-instantiations wrt to the $A$-variables (resp. to the $B$-variables) of the clauses~\eqref{eq:elwr},\eqref{s4} corresponding to the literals from $A_1$ (resp. from $B_1$) - for the iterated $\diff$-atoms introduced by the various steps of the algorithm, our construction ensures that the $0$-instantiations of the clauses~\eqref{s4} 
are true. Recall that $x$ is different  from $i^{\cM^+}$ for every index variable $i$ occurring in $A$ and recall that $y$ `duplicates' $x$, in the sense explained by (i)-(iv) above. This ensures that any instance of the clause~\eqref{eq:elwr} by an $A$-variable or by a common variable introduced during the algorithm,   relatively to  an atom $a_1=wr(a_2,i, e)$ occurring in $A_1$ must hold in $\cM^+$ (given that the analogous statement was true in $\cM$). A similar argument applies to
a formula
$j_1= \diff_1(a_1, a_2)\wedge \cdots \wedge j_r=\diff_r(a_1, a_2)$ occurring in $A_1$:~\footnote{ Recall that, according to the definition of a separated pair, if a separated pair contains $j_r=\diff_r(a_1, a_2)$, it must also contain some $j_s=\diff_s(a_1, a_2)$ for all $s<r$).} in fact, $y$ is different from $j_1^{\cM^+}, \dots, j_r^{\cM^+}$ and so either $a_1^{\cM^+}(y)=a_1^{\cM^+}(x)=a_2^{\cM^+}(x)=a_2^{\cM^+}(y)$ or both $x,y$ are below $j_r^{\cM^+}=j_r^{\cM}$. Checking the above facts for the formulae from $B_1$ is perfectly analogous.

Continuing as above, we produce a chain of $TO\cup\EUF$-models  witnessing the fact that we can run infinitely many steps of the algorithm without finding an inconsistency. 
\eop
 \end{proof}

 \subsection{Computing interpolants in $TO\cup \EUF$}\label{subsec:complexity}

In this section,
we show how to compute $TO\cup \EUF$-interpolants by exploiting well known interpolation procedures from the  literature of theory combination in SMT. We will then discuss the overall complexity in time of the  $TO\cup \EUF$- and $\AXD$-interpolation procedures.

Our approach, which is motivated by efficient implementations in the state-of-the-art SMT solvers, relies on the use of Yorsh and Musuvathi \cite{yorsh}  method: such a method  
works correctly for theories which (besides having quantifier-free interpolation) have disjoint signatures, are stably infinite, convex and \emph{equality interpolating}.~\footnote{ The convex hypothesis has been removed in~\cite{BGR14} at the cost of a more complicated algorithm and a more complicated definition of being equality interpolating. Recall that a convex theory is said to be 
equality interpolating  iff 
     for every pair $y_1, y_2$ of variables and for every pair
       of constraints $\delta_1(\ux,\uz_1, y_1),
       \delta_2(\ux,\uz_2,y_2)$ such that
       \begin{equation*}
         \label{eq:ym_ant}
         T\vdash\delta_1(\ux, \uz_1,y_1)\wedge
         \delta_2(\ux,\uz_2, y_2)\to
         y_1= y_2
       \end{equation*}
       there exists a term $t(\ux)$ such that
       \begin{equation*}
         \label{eq:ym_cons}
         T\vdash
         \delta_1(\ux, \uz_1, y_1)\wedge
         \delta_2(\ux, \uz_2, y_2)\to
         y_1= t(\underline{x})  \wedge  y_2= t(\underline{x}).
       \end{equation*}
} This is the case of our two theories $TO$ and $\EUF$: the equality interpolating condition is equivalent to strong amalgamation~\cite{BGR14} in presence of quantifier-free interpolation; recall also that $TO$ is convex as a subtheory of linear real arithmetic.

In order to compute the overall complexity of the combined procedure, we now recall the cost of the method from \cite{yorsh}. The main sources of complexity are given (i) by the complexity costs of computing interpolants in the component theories 
and (ii) by the additional steps required by the combination method. 
The former will be discussed in the next subsection, where it will be shown that the cost (i) for arbitrary quantifier-free formulae is exponential for both $TO$ and $\EUF$ (this is not surprising, given that  all algorithms operating in the Boolean propositional case are already worst-case exponential,
see also the observations in~\cite{mundici}).
The costs mentioned in (ii) rely on (ii1) the cost of the exchange of equalities among variables between the two component theories via the Nelson-Oppen method and on (ii2) the cost of computing the equality interpolating terms for such equations in order to guarantee their pureness. 

It is well known that  Nelson-Oppen satisfiability check for conjunctions of literals 
requires polynomially many calls to analogous problems in the component theories.
 In addition, it can be easily seen that also the computation of the equality interpolating terms can be done in polynomial time. Indeed, in case of $TO$,  equality interpolating terms are trivial because they can only be variables. Moreover, in case of $\EUF$, it is well known from \cite{yorsh} that a decision procedure for $\EUF$ can be easily modified to generate only pure equalities. The idea is to modify the implementation of the congruence closure
algorithm \cite{kapurCC} to choose a representative for an equivalence class to be an $AB$-common term (i.e., a term whose variables are shared by both $A$ and $B$), whenever an equivalence class contains at least one such term.
When an equivalence class contains both $A$-local and $B$-local terms (i.e., respectively, terms whose variables are in $A$ but not in $B$ and vice versa), it is easily seen, by induction on the congruence closure manipulations, that it also contains an $AB$-common term. This can be of course computed in polynomial time as well.

In conclusion, the overall complexity of computing $TO\cup \EUF$-interpolants has an \emph{exponential upper bound} in time; the same upper bound works for the computation of $\AXD$-interpolants because Theorem~\ref{thm:allin} shows that $\AXD$-interpolants can be computed by a polynomial call to a $TO\cup \EUF$-interpolants computation problem.  





\subsection{Implementation issues}

One of the reasons why interpolation is so attractive from the applications point of view is that, despite the above mentioned costs, it is usually possible to extract interpolants in \emph{linear time} from a given refutation proof, thus the computational costs are moved to the costs of finding a refutation. This is the reason why in the literature, interpolation has been often coupled with the study of appropriate calculi. In our case, a specific calculus for $TO\cup \EUF$ seems not to have been developed. Here we show how to adapt to our framework some existing calculi for stronger theories. 

We already described the Yorsh-Musuvathi combination method from~\cite{yorsh}; since we apply it to the comnined theory $TO\cup \EUF$, it remains to analyze here the availability and the cost of the input interpolation algorithms for both $TO$ and $\EUF$.

Let us first examine the case of $\EUF$. This is largely covered by the literature, see e.g.~~\cite{McMillan05}. Since in our application we only have unary function symbols, one can also efficiently compute interpolants via some recent algorithm producing \emph{uniform interpolants}~\cite{GGK,cade19}: in fact, for the case where function symbols are all unary, uniform interpolants can be extracted in quadratic time~\cite{GGK} from a conjunction of literals (extracting them from an arbitrary quantifier-free formula requires a DNF conversion, 
to be possibly handled via efficient structural subformula renaming transformations). 


The case of $TO$ requires some better investigation. In order to be able to re-use existing algorithms and tools, it is useful to notice that $TO$ 
has the same universal fragment as many other stronger theories like 
$\LIA$, $\IDL$, $\LRA$,
etc.\footnote{To realize that such theories have the same universal fragment as $TO$ (that they are 'co-theories' of $TO$, in the terminology of~\cite{CK}), it is sufficient to show that every model of $TO$ embeds into a model of such theories. This can be done easily  by compactness and diagrams, see~\cite{CK}.} Thus whenever we need a satisfiability test in
$TO$
for a quantifier-free formula, we can freely perform it in one of the above mentioned stronger theories. The same is true for interpolation, however notice that the language of 
$TO$ is poorer, so less interpolants are available (but at least one always exists because $TO$ is amalgamable, even strongly amalgamable!). Thus, if we want to re-use in our context an available interpolating prover for one of the above mentioned richer theories, we must check in advance that the algorithm underlying such a prover \emph{does not introduce spurious symbols} (typically the sum symbol) when it computes quantifier-free interpolants in the language of 
$TO$.

%

We make use of the procedure described in \cite{cimatti-acm} for $\IDL$ (integer difference logic): 
such procedure uses very efficient solving algorithms based on graph-based decision procedures \cite{cotton,cimatti-acm}.
The technique has been  implemented within the MATHSAT 4 SMT solver \cite{brutto08}.

Our input formula is in the signature of $TO$ and   we want to reduce our problem to the problem  of computing $\IDL$-interpolants of conjunctions of atoms of the kind $x\leq y$ or of the form $x<y$. This requires a rewriting of our literals as explained below and a subsequent  DNF-conversion:
\begin{compactenum}
\item[-] atoms in $TO$ of the form $x\leq y$, thought as atoms in $\IDL$, are rewritten to $0\leq y-x$.
\item[-] atoms in $TO$ of the form $x< y$, thought as atoms in $\IDL$, are rewritten to $0\leq y-x-1$.
\item[-] atoms in $TO$ of the form $x=y$, thought as atoms in $\IDL$, are rewritten to the conjunction $0\leq y-x \land 0\leq x-y$.
\item[-]  negated atoms in $TO$ of the form $x\neq y$, thought in $\IDL$, are rewritten to the disjunction $x< y\lor y<x$, i.e., to $0\leq y-x-1\lor 0\leq x-y-1$.
\end{compactenum}
Thus we need to analyze the interpolation procedure for conjunctions of \emph{difference bounds}, i.e. of atoms of the kind $0\leq y-x +c$ (being aware of the fact that such atoms come from rewriting of $TO$-atoms only when $c=0,-1$).

It is well-known that checking the consistency of conjunctions of difference bounds amounts to inspecting the existence of cycles with negative weights in the graph induced by them. Indeed, the difference bound $0\leq y-x +c$  can be interpreted as an edge $x\to^c y$ with weight $c$ from  the vertex $x$ to the vertex $y$. It can be easily proved that a set of difference bounds is inconsistent if and only if this induced graph has a cycle of negative weight.

The procedure in \cite{cimatti-acm} tackles the problem of searching interpolants for two conjunctions of difference bounds   $A$ and $B$ such that $A\land B$  is $\IDL$-inconsistent, by computing first  a negative cycle $\kappa$ in the graph corresponding to $A\land B$ (this can be detected using efficient consistency check algorithms as the one in \cite{cotton}). In the following, by abuse of notation, we will use the notation $A$, $B$, $\kappa$ etc.  for denoting the corresponding induced graphs as well. If $\kappa \subseteq A$, then $A$ is inconsistent, in which case the interpolant is $\bot$. Similarly, when
$\kappa \subseteq B$, the interpolant is $\top$.

 If neither of these occurs, then the edges in the cycle can be
partitioned in subsets of $A$ and $B$.
 We call maximal $A$-path inside $\kappa$, i.e., paths of the form $x_1\to^{c_1}\dots \to^{c_n} x_n$ such that (i) $x_i\to^{c_i} x_{i+1} \in  A $ for $i \in [1, n - 1]$, and (ii) $\kappa$ contains $x'\to^{c'}x_1$  and $x_n\to^{c''}x''$ that are in $B$ as well. In the non-trivial case, the procedure searches for these maximal $A$-path. Intuitively, it means that it looks for maximal subpaths of $\kappa$ included in $A$ such that the endpoints are in $B$: this in particular implies that the endpoints are common variables of $A$ and $B$. 
 
 Let the \emph{summary constraint} of a maximal $A$-path
$x_1\to^{c_1}\dots \to^{c_n} x_n$  be the inequality $0 \leq x_n - x_1 + \Sigma_{i=1}^{n-1} c_i$. 
The outcome of the algorithm is the  following: 
\begin{compactenum}
\item[-] $\bot$, if $\kappa \subseteq A$;
\item[-] $\top$, if $\kappa \subseteq B$;
\item[-] otherwise, the conjunction of the summary constraints of all the maximal $A$-paths inside $\kappa$.
\end{compactenum}

It can be easily shown that the output above is an
interpolant for $A$ and $B$. The first two cases are trivial. Regarding the third case, it is evident that the proposed interpolant is implied by $A$ and inconsistent with $B$ (the latter happens because the graph built up from the conjunction of $B$ and of the interpolant still contains a negative cycle - it will be the negative cycle obtained from $\kappa$ by contracting maximal $A$-paths into single edges).
 
Notice that, concerning difference bounds rewritten from $TO$-atoms,  we  have the following two cases for summary constraints:

\begin{compactenum} 
\item[-] if no strict inequality is involved, then $\Sigma_{i=1}^{n-1} c_i=0$ and so the summary constraint is $0\leq x_n-x_1$, which is $\IDL$-equivalent to $x_1\leq x_n$;
\item[-]  if there is some strict inequality involved, then $0<\Sigma_{i=1}^{n-1} c_i\leq n-1$ and we get that the summary constraint is $0 \leq x_n - x_1 - (\Sigma_{i=1}^{n-1} c_i)$, which can be weakened to
 $0 \leq x_n - x_1 - 1$ (since $-(\Sigma_{i=1}^{n-1} c_i )< 0$); the latter is $\IDL$-equivalent to
$x_1<x_n$  .
\end{compactenum}
Notice that the summary constraints, even weakened as above, are still sufficient to produce an inconsistency with $B$, because they still allow to build a negative cycle from $\kappa$ by contracting maximal $A$-paths into single negatively-weighted edges. Hence, the $\IDL$-interpolant, suitably weakened as shown above,
produces an interpolant 
in the restricted language of linear orders. 

Regarding the complexity, the only non linear cost of the above procedure is the identification of the negative cycle $\kappa$, which can be still performed in polynomial time  though.\footnote{The cost mentioned in~\cite{cotton} is 
 $O(m+ n\cdot log\,n)$, where $n$ is the  number of variables and  $m$ the number of difference bounds. 
 }
One should however recall that the preprocessing step producing a DNF built up from  difference bounds requires exponential time.
%

The conclusion of the complexity analysis of this Section is that \emph{computing quantifier-free interpolants in $TO\cup \EUF$ (and hence also in $\AXD$) requires exponential time}.  However, we identified a \emph{special class of formulae} (namely conjunctions of $\EUF$-literals and of difference bounds) where computation of interpolants is polynomial. Such  special class can be of interests for applications, because often infinite state model-checkers (like \textsc{Booster} and its underlying engine \textsc{mcmt}) use it to represent sets of reachable states. 

A tool (called \emph{AXDInterpolator} \cite{Jose}) for computing $\AXD$-interpolants in case the index theory is $TO$ is currently in construction. This project implements an interpolation algorithm for the theory of arrays extended with the diff operator by computing a reduction from the $\AXD$ theory to the theory of $TO\cup\EUF$ and calling an interpolation engine to process the reduced formula. Currently, the software supports iZ3 and Mathsat as such engines.

\subsection{Uniform interpolants}

Uniform interpolants, concerning especially $\EUF$, received special attention in recent literature~\cite{kapur,kapurJSSC,GGK,cade19,CILC20}. We show here  that uniform interpolants do not exists in \AXD, using a model-theoretic argument. First, we  recall the involved definition.

Fix a  theory $T$ and an existential formula $\exists \ue\, \phi(\ue, \uz)$; call a \emph{residue} of $\exists \ue\, \phi(\ue, \uz)$ any quantifier-free formula belonging to the set of quantifier-free formulae  
$$
Res(\exists \ue\, \phi) ~=~\{\theta(\uz, \uy)\mid T \models \phi(\ue, \uz) \to \theta(\uz, \uy)\}.
$$ 
A quantifier-free formula $\psi(\uz)$  
is said to be a \emph{$T$-uniform interpolant} (or, simply, a \emph{uniform interpolant}, abbreviated UI) of $\exists \ue\, \phi(\ue,\uz)$ iff  $\psi(\uz)\in Res(\exists \ue\, \phi)$ and $\psi(\uz)$ implies (modulo $T$) all the other formulae in $Res(\exists \ue\, \phi)$. It is immediately seen that  UI are unique (modulo $T$-equivalence).
We say that a theory $T$ has \emph{uniform quantifier-free interpolation} iff every existential formula $\exists \ue\, \phi(\ue,\uz)$  
has a UI. It is clear that if $T$ has uniform quantifier-free interpolation, then it has ordinary quantifier-free interpolation: in fact, if $T$ has uniform quantifier-free interpolation, then there is an interpolant $\theta$ which  can be used as interpolant for all entailments $T\models \phi(\ue, \uz)\to \phi'(\uz, \uy)$, varying the quantifier-free formula $\phi'$. 

UI are semantically characterized by the following result, taken from~\cite{cade19}).

\begin{lemma}[Cover-by-Extensions]\label{lem:cover} A formula $\psi(\uy)$ is a UI in $T$ of $\exists \ue\, \phi(\ue, \uy)$ iff 
it satisfies the following two conditions:
\begin{description}
\item[{\rm (i)}] $T\models  \forall \uy\,( \exists \ue\,\phi(\ue, \uy) \to \psi(\uy))$;
\item[{\rm (ii)}] for every model $\cM$ of $T$, for every tuple of  elements $\ua$ from the support of $\cM$ such that $\cM\models \psi(\ua)$ it is possible to find
  another model $\cN$ of $T$ such that $\cM$ embeds into $\cN$ and $\cN\models \exists \ue \,\phi(\ue, \ua)$.
\end{description}
\end{lemma}

\begin{theorem}
  \AXD does not have uniform quantifier-free interpolation.
\end{theorem}

\begin{proof}
 We show that there does not exist a uniform interpolant of the formula
 $$
 i_1<i_2 \land i_2< i_3 \land rd(c_1,i_2) \neq rd(c_2, i_2)
 $$
 with respect to $i_2$ (i.e. $i_2$ is the existential variable that should be `eliminated' in a uniform way).\footnote{
 This is the formula $B$ from Example~\ref{ex2}; as in Example~\ref{ex2}, variables $i_1, i_3, c_1, c_2$ are `common' and are not eliminated. 
 }
 We use the 'cover-by-extension' Lemma~\ref{lem:cover} above. Suppose that such a uniform interpolant $\phi(c_1,c_2, i_1, i_3)$ exists. Consider the following sequence of full functional  models 
 $\cM_i$ ($i\geq 2$) of $\AXD$. The sort \ELEM is interpreted as the two-element set $\{0,1\}$ in all $\cM_i$ (with $\bot$ being equal to 1).
 The sort \INDEX is interpreted as the integer interval $[0, i]$ in $\cM_i$;  $c_1, c_2$ are  functions differing from each other at all indices bigger than 1; $i_1$ is 0 and $i_3$ is 1. Now  in $\cM_i$ the sentence 
\begin{equation}\label{eq:ui}
 \exists i_2~(i_1<i_2 \land i_2< i_3 \land rd(c_1,i_2) \neq rd(c_2, i_2))
\end{equation}
is false and it remains false in every superstructure of $\cM_i$, because $c_2$ can be obtained from $c_1$ by iterated writing operations via indexes $0, \dots, i$ and such a fact (being expressible at quantifier-free level) must hold in superstructures too: this entails that $c_1$ and $c_2$ must agree on extra indices introduced in superstructures. 
Thus, according to the cover-by-extension lemma, $\cM_i\models \neg \phi$.

Consider now an ultraproduct $\Pi_D \cM_i$ modulo a non-principal ultrafilter $D$. By the fundamental \L{os} theorem~\cite{CK}, we have $\Pi_D \cM_i\models \neg \phi$.  If we manage to extend $\Pi_D \cM_i$ to a superstructure $\cN'$ where~\eqref{eq:ui} holds, we get a contradiction by the cover-by-extension lemma.

Notice that in $\Pi_D \cM_i$, the relationship $c_1\sim_{\Pi_D \cM_i} c_2$ does not hold because now $c_1$ and $c_2$ differ on infinitely many indices, so we can enlarge $\INDEX^{\Pi_D \cM_i}$ by adding it a new index $i_2$ between $i_1$ and $i_2$ and letting $c_1(i_2)$ be different from $c_2(i_2)$. Formally, we define the following superstructure $\cN$ of $\Pi_D \cM_i$: the sort $\ELEM$ is interpreted as in $\Pi_D \cM_i$, whereas the sort \INDEX is $\INDEX^{\Pi_D \cM_i}\cup\{i_2\}$ with 
$0<i_2<1$ (here $0$ and $1$ are the equivalence classes modulo $D$ of the constant functions with values $0$ and $1$, respectively). The sort $\ARRAY^\cN$ contains all functions from 
$\INDEX^{\cN}$ to $\ELEM^{\cN}$ (thus the model $\cN$ is full). We extend the functions from 
 $\ARRAY^{\Pi_D \cM_i}$\footnote{Recall that every model of $\AXD$ is isomorphic to a functional model, albeit not to a full one.}  to the new index set as follows. We let $d(i_2)=0$ iff 
$d\sim_{\Pi_D \cM_i} c_1$ and $d(i_2)=1$, otherwise.\footnote{Here we suppose that 
 $c_1\not\sim_{\Pi_D \cM_i} \epsilon$, otherwise we swap $c_1$ and $c_2$.
 } This extension preserves $rd$ and $wr$  operations. It preserves also $\diff$: in fact, if $\diff^{\Pi_D \cM_i}(d_1,d_2)=k>0$, then 
$\diff^{\cN}(d_1,d_2)$ is defined and 
$\diff^{\Pi_D \cM_i}(d_1,d_2)=\diff^{\cN}(d_1,d_2)$ and if  $\diff^{\Pi_D \cM_i}(d_1,d_2)=0$, then 
$d_1\sim_{\Pi_D \cM_i} d_2$, so again $\diff^{\cN}(d_1,d_2)$ is defined and $\diff^{\Pi_D \cM_i}(d_1,d_2)=\diff^{\cN}(d_1,d_2)$.
Thus $\cN$ is a superstructure of $\Pi_D \cM_i$ where ~\eqref{eq:ui} holds. $\cN$ can be extended to a model $\cN'$ of $\AXD$ via Lemma~\ref{lem:extension} and this concludes our proof.
\eop
\end{proof}

\end{document}